\documentclass[11pt]{article}
\usepackage[utf8]{inputenc}
\usepackage[english]{babel}
\usepackage[active]{srcltx}
\usepackage{hyperref}
\usepackage[all]{xy}
\usepackage{amsfonts}
\usepackage{amsthm}
\usepackage{enumerate}
\usepackage{mathrsfs}
\usepackage{multicol}
\usepackage[all]{xy}
\usepackage{amsfonts}
\usepackage{latexsym,amsmath,mathtools,amsthm,amssymb}
\usepackage{lmodern}
\usepackage[T1]{fontenc}
\usepackage[usenames,dvipsnames]{color}
\usepackage{xfrac}
\usepackage{pstricks}
\usepackage[absolute,overlay]{textpos}
\usepackage{graphics,wrapfig,times}

\usepackage{color}
\definecolor{DarkBlue}{rgb}{0.1,0.1,0.5}
\definecolor{Red}{rgb}{0.9,0.1,0.1}
\definecolor{Green}{rgb}{0.3,0.7,0.0}
\definecolor{green2}{rgb}{0.1,0.7,0.2}
\definecolor{blue2}{rgb}{0.0,0.6,0.7}
\definecolor{pink}{rgb}{1,0.0,1}
\definecolor{orange}{rgb}{0.9,0.0,0.1}

\newtheorem{theo}{Theorem}

\newtheorem{lemma}{Lemma}
\newtheorem{prop}{Proposition}

\newtheorem{definition}{Definition}

\newcommand{\ep}{\varepsilon}

\def\sode{\textsc{sode }}
\def\sopde{\textsc{sopde }}

\newcommand{\R}{\mathbb{R}}      

\renewcommand{\d}{\mathrm{d}}
\newcommand{\derpar}[2]{\displaystyle\frac{\partial{#1}}{\partial{#2}}}

\newcommand{\Lag}{\mathcal{L}}
\newcommand{\Leg}{\mathcal{FL}}
\newcommand{\vf}{\mathfrak{X}}
\newcommand{\df}{\Omega}

\newcommand{\Tan}{\mathrm{T}}
\newcommand{\inn}{{\mathop{i}\nolimits}}
\newcommand{\Lie}{\mathop{\mathrm{L}}\nolimits}

\newcommand{\bal}{\begin{align*}}
\newcommand{\eal}{\end{align*}}
\def\beq{\begin{equation}}
\def\eeq{\end{equation}}
\def\bea{\begin{eqnarray}}
\def\eea{\end{eqnarray}}
\def\beann{\begin{eqnarray*}}
\def\eeann{\end{eqnarray*}}
\def\ben{\begin{enumerate}}
\def\een{\end{enumerate}}
\def\bit{\begin{itemize}}
\def\eit{\end{itemize}}
\def\ds{\displaystyle}

\def\dst{\displaystyle}
\def\ep{\vspace{5pt}}

\def\vf{\mathfrak X}
\def\df{{\mit\Omega}}
\def\Lag{{\cal L}}

\def\d{{\rm d}}

\def\R{\mathbb{R}}

\def\cL{\mathcal{L}}

\def\cF{\mathcal{C}^\infty}
\def\cE{\mathcal{E}}
\def\cS{\mathcal{S}}

\def\x{\textnormal{q}}
\def\v{\textnormal{v}}

\def\bX{\mathbf{X}}
\def\bY{\mathbf{Y}}

\def\bV{\mathbf{V}}
\def\bG{\mathbf{\Gamma}}
\def\bD{\mathbf{\Gamma}}
\def\bOm{\mathbf{\Omega}}

\def\fX{\mathfrak{X}}
\def\fM{\mathfrak{M}}

\def\Tan{{\rm T}}
\def\Lie{\mathop{\rm L}\nolimits}
\def\inn{\mathop{i}\nolimits}
\def\Cinfty{{\rm C}^\infty}

\def\Ker{\textnormal{Ker }}

\title{\sc The second-order problem for $k$-presymplectic Lagrangian field theories.
Application to the Einstein--Palatini model}
\author{\sffamily 
\sc $^a$David Adame--Carrillo
\thanks{david.adamecarrillo@aalto.fi.} ,
$^b$Jordi Gaset
\thanks{jordi.gaset@uab.cat\,({\it ORCID}:\,0000-0001-8796-3149).} ,
$^c$Narciso Rom\'an--Roy
\thanks{narciso.roman@upc.edu\,({\it ORCID}:\,0000-0003-3663-9861).} .
\\[1ex]
\normalsize\itshape\sffamily 
$^a$Department of Mathematics and Systems Analysis, Aalto University, Espoo, Finland.
\\[1ex]
\normalsize\itshape\sffamily 
$^b$Department of Physics,
Universitat Aut\`onoma de Barcelona,
Bellaterra, Spain.
\\[1ex]
\normalsize\itshape\sffamily 
$^c$Department of Mathematics,
Universitat Polit\`ecnica de Catalunya,
Barcelona, Spain.
}

\date{February 1, 2022}

\begin{document}

\maketitle

\pagestyle{myheadings}

\thispagestyle{empty}

\begin{abstract}
In general, the system of $2$nd-order partial differential equations made of the Euler--Lagrange equations of classical field theories are not compatible for singular Lagrangians. This is the so-called {\sl second-order problem}. The first aim of this work is to develop a fully geometric constraint algorithm which allows us to find a submanifold where the Euler-Lagrange equations have solution,
and split the constraints into two kinds depending on their origin. We do so using $k$-symplectic geometry, which is the simplest intrinsic description of classical field theories. 
As a second aim, the {\sl Einstein--Palatini} model of General Relativity 
is studied using this algorithm.
\end{abstract}



\vbox{\raggedleft AMS s.\,c.\,(2020): 
{\it Primary\/}: 53D42, 70S05, 83C05. 
{\it Secondary\/}: 53C15, 35Q76, 53Z05, 55R10.}\null

\markright{{\rm D. Adame--Carrillo, J. Gaset, N. Rom\'an--Roy},
    {\sl $k$-presymplectic field theories. Einstein--Palatini model.}}


\setcounter{tocdepth}{2}
{
\small
\def\addvspace#1{\vskip 1pt}
\parskip 0pt plus 0.1mm
\tableofcontents
}


\section{Introduction}

It is an established fact that {\sl symplectic geometry} is the most suitable
geometric framework to describe Lagrangian and Hamiltonian (autonomous) mechanics \cite{AM-fm,Ar,LM-sgam}.
As classical field theories appear many times as a generalisation of mechanics, 
it seems very natural to try to describe classical field theories with a generalization of symplectic geometry.
One of the most generic and complete approaches is the {\sl multisymplectic description},
where jet bundles and bundles of forms are used as the manifolds where 
the Lagrangian and the Hamiltonian formalisms take place 
(see, for instance, \cite{RomRoy} and the references therein).
However, the simplest approach is the {\sl $k$-symplectic description}
\cite{Awane,mt2,Gu-87},
which is used to describe systems in field theory whose Lagrangian and Hamiltonian  functions
depend only on the fields and derivatives of them, 
or their associate momenta in the Hamiltonian formulation. 
In the Lagrangian formalism, this formulation takes place 
in the $k$-tangent bundle of some manifold $Q$, which is denoted $\Tan^1_kQ$ and,
in the Hamiltonian formalism, in the $k$-cotangent bundle $(\Tan^1_k)^*Q$. 
These bundles are the Whitney sum of $k$ copies of the tangent bundle $\Tan Q$ and the cotangent bundle $\Tan^*Q$, respectively \cite{dLSV,fam,RRS}.

The $k$-symplectic framework is well suited to dealing with regular Lagrangians. 
Nevertheless, many of the relevant models in physics feature singular Lagrangians
and this leads to work in the so-called {\sl $k$-presymplectic} framework. 
Singular systems are important  because some of 
the most important physical theories are singular; for instance, 
Maxwell's electromagnetism, all the models in General Relativity, 
string theory and gauge theories in general. 
The main problem of these singular theories is the failure of the 
usual theorems for the existence of solutions 
of the differential equations which describe them. 
This problem is usually solved by applying suitable {\sl constraint algorithms}
which allow us to find a submanifold of the phase space 
of the system where the existence of solutions is assured.

 The first of these constraint algorithms was given by P.G. Bergmann and P.A.M. Dirac, using a local coordinate language, 
 for the Hamiltonian formalism of singular mechanics \cite{AB-51,Di-50}. 
 Geometric versions of this algorithm were developed later, 
 both for the Hamiltonian and Lagrangian formalisms of autonomous mechanical systems \cite{got79,got78,GP-92,MMT-97} and also 
 for non-autonomous systems \cite{chi94,dLe02,GM2005}.
 Furthermore, the problem of the compatibility of the Hamiltonian field equations ($1$st-order PDE’s) for singular field theories has already been solved 
 in the (pre)multisymplectic \cite{dLe96B,dLe05} and the $k$-presymplectic and $k$-precosymplectic frameworks \cite{GMR,GRR}. 

 One of the characteristic features of the Lagrangian formalism is that 
 physical and variational motivations demand that 
 the equations that describe the behaviour of the system 
 must be ordinary second-order differential equations (\textsc{sode}), 
 in the case of mechanics, and second-order partial differential equations
 (\textsc{sopde}), in field theories: the {\sl Euler--Lagrange equations}. 
 In the above mentioned geometric descriptions of mechanics and field theories, 
 solutions to the equations of the system are represented by vector fields and $k$-vector fields, respectively. 
 Then, the physical solutions are the integral curves or the integral sections to these vector and $k$-vector fields, 
 which verify the Euler--Lagrange equations. 
 But, in order to assure this last fact, the vector and $k$-vector fields must fulfil that 
 their integral curves and sections must be {\sl holonomic}; that is, canonical liftings of curves and maps to the bundles where they are defined
 \cite{Crampin,SCC-84,Sa-89}. 
 This is the so-called {\sl second-order condition} and the vector fields and $k$-vector fields fulfilling this condition 
 are called \textsc{sode}'s and \textsc{sopde}'s, respectivelly.

 In the regular case  the second-order condition holds for every solution to the geometrical equations of the system; 
 whereas in the case of singular Lagrangian systems these solutions, if they exist, 
 do not satisfy this condition, in general, 
 and hence it is an additional problem to study, 
 besides the existence of solutions. 
 One of the first geometric analysis of this problem 
 for singular Lagrangian (autonomous) mechanics was done in \cite{got80}, 
 where the authors found a submanifold of the velocity phase space 
 where a {\sc sode} solution exists; although, in general, this is not a maximal submanifold
 and the procedure to obtain it is not algorithmic. 
 A complete and constructive algorithm for finding a maximal submanifold was developed in \cite{BGPR,BGPR2}, 
 but the treatment is local-coordinate. 
 The procedure was intrinsically reformulated (partially) later in \cite{CLR,CLR2}, 
 but a complete geometric algorithm is given, for the first time, in \cite{MR} where, given a singular Lagrangian system $(\Tan Q;\cL)$, 
 an algorithm that finds the maximal submanifold of $\Tan Q$ on which one can find \sode solutions was developed. 
 Later, this algorithm was extended to singular non-autonomous Lagrangians \cite{dLe02}.

 Nevertheless, an equivalent geometric algorithm for solving 
 the second-order problem for singular Lagrangian field theories 
 is not known yet. 
 The first part of this work is an attempt to solve the problem for singular Lagrangians in the $k$-symplectic formulation. 
 Thus, given a singular Lagrangian in $\Tan^1_kQ$,
 we provide an algorithm that  allows us to find the maximal submanifold of 
$\Tan^1_kQ$ on which we can find \sopde solutions to the Lagrangian field equations.
In addition, the emergent constraints are classified into two groups,
depending on whether they are a consequence of the compatibility of field equations 
or the requirement that solutions verify the second order condition.
This algorithm is the generalization of what is presented in \cite{MR}
for singular Lagrangian mechanics to $k$-presymplectic Lagrangian systems,
and completes the algorithm presented in \cite{GMR}, where the second-order
problem was not considered.

As a very interesting application of this method we study
the $k$-symplectic description of the {\sl Einstein--Palatini Lagrangian model} 
of General Relativity \cite{Einstein,Pa-19,FFR-82}.
This system, which is also known as the {\sl metric-affine model},
consists in considering the {\sl Hilbert Lagrangian} for the {\sl Einstein equations},
but taking an arbitrary connection instead of the Levi--Civita connection
associated with the metric.
The resulting Lagrangian is affine and then singular;
thus, as a previous step, it is very useful to make a preliminary analysis on the characteristics of the affine Lagrangians in general.
In particular, the constraint algorithm for the Einstein-Palatini Lagrangian
gives different kinds of constraints, and
the results obtained here are discussed and compared with those obtained
in the multisymplectic description of this model \cite{GR-2018}.

The organisation of the paper is the following: 
Section 2 is devoted to review the main features on the $k$-symplectic approach to Lagrangian and Hamiltonian field theories. 
In Section 3, we present the generalisation of the constraint algorithm 
to $k$-symplectic singular Lagrangian systems. 
Finally, Section 4 is devoted to apply this method to
affine Lagrangians in general and,
in particular, to the Einstein--Palatini model of General Relativity.

Throughout the text we use the summation convention for repeated crossed tensorial indices. All the manifolds are real, second countable and $\Cinfty$.
Manifolds and mappings are assumed to be smooth.

\section{$k$-symplectic field theories}

In this Section we review the $k$-symplectic description of the
Lagrangian formalism of classical field theory
following the presentation given in \cite{dLSV}.

\subsection{$k$-symplectic geometry, $k$-tangent bundles and geometric structures}

\begin{definition} Let $M$ be a manifold of dimension $n(k+1)$,
$(V;\pi)$ a $nk$-dimensional integrable distribution
and $\omega^1,\ldots,\omega^k$ a family of closed differentiable 
$2$-forms on $M$. We say that $(\omega^1,\ldots,\omega^k;V)$ is a \textbf{$k$-symplectic structure} on $M$ if
\begin{enumerate}[(i)]
\setlength{\itemindent}{+1cm}
\item $\omega^\alpha\big|_{V\times V}=0$, $1\leq\alpha\leq k$,
\item $\bigcap_{\alpha=1}^{k}\Ker \omega^\alpha=0$.
\end{enumerate}
Then $(M;\omega^1,\ldots,\omega^k;V)$ is said to be \textbf{$k$-symplectic manifold}.
If some of the conditions in the above definition are not satisfied then 
$(M;\omega^1,\ldots,\omega^k,V)$ is called a \textbf{$k$-presymplectic manifold} 
and, similarly, $(\omega^1,\ldots,\omega^k)$ is a \textbf{$k$-presymplectic structure} on $M$.
\end{definition}

A simple example of a $k$-presymplectic manifold is a submanifold of a $k$-symplectic manifold, where the pull-back of the $2$-forms of the $k$-symplectic structure by the embedding map yields the $2$-forms of the $k$-presymplectic structure in the submanifold.
Furthermore, the canonical model for $k$-symplectic manifolds is
the \emph{\textbf{$k$-cotangent bundle}} of a manifold $Q$,
which is the Whitney sum of $k$ copies of the cotangent bundle $\Tan^*Q$; that is,
$(\Tan^1_k)^*Q\coloneqq\Tan^*Q\oplus_{{ }_Q} \overset{k}{\dotsb}\oplus_{{ }_Q} \Tan^*Q$,
which is endowed with the natural projections
$$
\tau\colon (\Tan^1_k)^*Q\longrightarrow Q
\quad , \quad
\tau^\alpha\colon (\Tan^1_k)^*Q\longrightarrow\Tan^*Q \ .
$$
Then, if $\omega_o\in\df^2(\Tan^*Q)$ denotes the canonical $2$-form in the cotangent bundle,
we can construct a canonical $k$-symplectic structure in $(\Tan^1_k)^*Q$
by taking $\omega^\alpha:=(\tau^\alpha)^*\omega_o$,
and being $V$ the vertical subbundle of $(\Tan^1_k)^*Q$ for the natural projection $(\Tan^1_k)^*Q\to Q$.\ep

Given a $n$-dimensional manifold $Q$, the \emph{\textbf{$k$-tangent bundle}} of $Q$ 
(which is also called the {\sl tangent bundle of $k^1$-velocities\/} of $Q$) 
is defined as the Whitney sum of $k$ copies of the tangent bundle $\Tan Q$; that is,
$\Tan^1_kQ\coloneqq\Tan Q\oplus_{{ }_Q} \overset{k}{\cdots}\oplus_{{ }_Q}\Tan Q$.
Points in $\Tan^1_kQ$ are denoted $(q;v_{q1},\ldots,v_{qk})\equiv{\bf v}_q$, where $q\in Q$ and $v_{q\alpha}\in\Tan_qQ$.
If $({\rm q}^i)$ are coordinates in $Q$, then natural coordinates in $\Tan Q$ 
are denoted $({\rm q}^i,\v^i_1,\ldots,\v^i_k)$, $1\leq i\leq n$.
The bundle $\Tan^1_kQ$ is endowed with the natural projections
$$
\pi\colon\Tan^1_kQ\to Q \quad , \quad
\pi^\alpha\colon\Tan^1_kQ\to\Tan Q \ .
$$

Given a map $\Phi\colon Q_1\longrightarrow Q_2$,
the {\sl $k$-extension} of $\Phi$ to $\Tan^1_kQ_1$
is the map $\Tan^1_k\Phi\colon\Tan^1_kQ_1\longrightarrow\Tan^1_kQ_2$ defined by
$\Tan^1_k\Phi(q;v_{q1}\ldots,v_{qk})\coloneqq (\Phi(q);\Tan_q\Phi(v_{q1}),\ldots,\Tan_q\Phi(v_{qk}))$.

In the same way, the {\sl prolongation} of a map $\psi:U\subseteq\R^k\longrightarrow Q$ to $\Tan^1_kQ$ is the map
\begin{equation*}
\begin{split}
\tilde{\psi}\ :\ U\subseteq\R^k\ & \longrightarrow\Tan^1_kQ\\
x\ & \longmapsto\ \tilde{\psi}(x)\coloneqq\Big(\psi(x);\psi_*\Big(\frac{\partial}{\partial t^1}\Big|_x\Big),\ldots,\psi_*\Big(\frac{\partial}{\partial t^k}\Big|_x\Big)\Big)
\end{split}
\end{equation*}
where $t^1,\ldots,t^k$ are the natural coordinates of $\R^k$.
Then $\tilde\psi$ is said to be {\sl holonomic}.

If $v\in\Tan_qQ$ and ${\bf v}_q\in\Tan^1_kQ$,
the {\sl vertical $\alpha$-lift} of $v$ to $\Tan_{{\bf v}_q}(\Tan^1_kQ)$ is defined by
the expression
$\ds (v)^{\wedge\alpha}_{{\bf v}_q}\coloneqq\frac{d}{dt}\bigg|_{t=0}\big(q;v_{q1},\ldots,v_{q\alpha}+tv,\ldots,v_{qk}\big)$.
In natural coordinates, if
$\ds v=v^i\frac{\partial}{\partial\x^i}\bigg|_{q}$, then 
$\ds (v)^{\wedge\alpha}_{{\bf v}_q}=
v^i\frac{\partial}{\partial\v^i_\alpha}\bigg|_{{\bf v}_q}$.

The \emph{\textbf{vertical endomorphisms}} of $\Tan^1_kQ$ are the
$(1,1)$-tensor fields $J^\alpha$ on $T^1_kQ$ given by
\begin{equation*}
\begin{split}
(J^\alpha)_{{\bf v}_q}\colon\Tan_{{\bf v}_q}(\Tan^1_kQ)\ & \longrightarrow\Tan_{{\bf v}_q}(\Tan^1_kQ)\\
{\bf u}\ \ \ \ \ \ \ & \longmapsto\ \big(\Tan_1^k\pi({\bf u})\big)^{\wedge\alpha}_{{\bf v}_q}
\end{split} \ .
\end{equation*}
In natural coordinates, they are given by
$\ds J^\alpha=\frac{\partial}{\partial \v^i_\alpha}\otimes d\x^i$.

The \emph{\textbf{Liouville vector field}} $\Delta$ in $\Tan^1_kQ$ is defined as
$\ds\Delta_{{\bf v}_q}\coloneqq\sum_{\alpha=1}^k(v_\alpha)^{\wedge\alpha}_{{\bf v}_q}$,
for ${\bf v}_q\in\Tan^1_kQ$.
We also define the vector fields $\Delta_\alpha$ by 
$(\Delta_\alpha)_{{\bf v}_q}\coloneqq(v_\alpha)^{\wedge\alpha}_{{\bf v}_q}$ so that $\Delta=\Delta_1+\cdots+\Delta_k$.
The Liouville vector field is the generator of dilatations, i.e., 
its flow is given by the curves $(q;tv_1,\ldots,tv_k)$.
In natural coordinates, we have
$\ds\Delta=\sum_{\alpha=1}^k\v^i_\alpha\frac{\partial}{\partial\v^i_\alpha}$.

\begin{definition} A \textbf{$k$-vector field} in $Q$ is 
a section of the canonical projection 
$\pi\colon\Tan^1_kQ\longrightarrow Q$; that is, a map
$\bX\colon Q\longrightarrow\Tan^1_kQ$ such that $\pi\circ\bX={\rm Id}_Q$.
\end{definition}

Equivalently, a $k$-vector field is determined by $k$ vector fields $X_1,...,X_k\in\fX(Q)$ defined by $\pi^{\alpha}\circ\bX=X_{\alpha}$. 
Thus we write $\bX=(X_1,\ldots,X_k)$.
The set of $k$-vector fields on $Q$ is denoted by $\fX^k(Q)$.

An {\sl integral section} $\psi$ of $\bX\in\fX^k(Q)$ passing through 
$q\in Q$ is a map $\psi:U\subseteq\R^k\longrightarrow Q$, with $0\in U$, 
such that $\psi(0)=q$, and $\ds\psi_*\bigg(\frac{\partial}{\partial t^\alpha}\Big|_x\bigg)=(X_\alpha)_{\psi(x)}$; for every $x\in U$.
A $k$-vector field $\bX\in\fX^k(Q)$ is {\sl integrable} 
if there exists an integral section of $\bX$ passing through each point $q\in Q$. Then, in coordinates, if
$\ds X_\alpha=(X_\alpha)^i\frac{\partial}{\partial\x^i}+\sum_{\beta=1}^k(X_\alpha)^i_\beta\frac{\partial}{\partial \v^i_\beta}$,
we have that $\psi$ satisfies the system of differential equations
$$
(X_\alpha)^i\circ\psi=\frac{\partial(\x^i\circ\psi)}{\partial t^\alpha},\ \ \ \ \ (X_\alpha)^i_\beta\circ\psi=\frac{\partial(\v^i_\beta\circ\psi)}{\partial t^\alpha} \ .
$$
(Notice that, despite their name, in the $k$-symplectic formulation
these maps $\psi$ are not sections of any bundle projection).

\begin{definition}
\label{secondpd}
A  \textbf{second-order partial differential equation} ({\sc sopde}) is a $k$-vector field $\bX$ in $T^1_kQ$
(that is, a section $\bX\colon \Tan^1_kQ\longrightarrow\Tan^1_k(\Tan^1_kQ)$ of the projection
$\pi_{\Tan^1_kQ}\colon\Tan^1_k(\Tan^1_kQ)\longrightarrow \Tan^1_kQ$ 
which is also a section of $T^1_k\pi\colon\Tan^1_k(\Tan^1_kQ)\longrightarrow\Tan^1_kQ$.
\end{definition}

In natural coordinates, a \textsc{sopde} $\bX=(X_1,\ldots,X_k)$ has the expression
\begin{equation}\label{sopde-local}
X_\alpha=\v^i_\alpha\frac{\partial}{\partial\x^i}+\sum_{\beta=1}^k(X_\alpha)^i_\beta\frac{\partial}{\partial \v^i_\beta} 
\quad ; \quad (X_\alpha)^i_\beta\in\Cinfty(T^1_kQ) \ .
\end{equation}
Note that $\bX$ is a \textsc{sopde} if, and only if, $\ds\sum_{\alpha=1}^kJ^\alpha(X_\alpha)=\Delta$.

If $\bX=(X_1 ,\ldots,X_k) $ is an integrable {\sc sopde}, then a map $\psi:\R^k \to \R^k\times T^1_kQ$, given by
$\psi(x)=(\psi^\alpha(x),\psi^i(x),\psi^i_\alpha(x))$, is an integral section
of $\bX$ if, and only if, its components are
solution to the system of second order partial differential equations
\begin{equation}\label{nn}
\psi^\alpha(x)=x^\alpha \quad , \quad \psi^i_\alpha(x)=\frac{\displaystyle\partial \psi^i}
{\displaystyle\partial x^\alpha}(x) \quad , \quad \frac{\displaystyle\partial^2 \psi^i}
{\displaystyle\partial x^\alpha \displaystyle\partial x^\beta}(x)=(X_\alpha)^i_\beta(\psi(x)) \ ,
\end{equation}
and thus, $\psi$ is a holonomic section.
Observe that if $\bX$ is integrable, from
(\ref{nn}) we deduce that $(X_\alpha)^i_\beta=(X_\beta)_\alpha^i$.
This fact justifies the name \sopde for these kinds of $k$-vector fields; 
although the second order equations only refer
really to their integrable sections and then to
integrable $k$-vector fields. When these $k$-vector fields satisfy the condition of the definition \ref{secondpd} but they are not integrable,
they are also called {\sl \textbf{semiholonomic $k$-vector fields}},
and \sopde $k$-vector fields which are integrable are called 
{\sl \textbf{holonomic $k$-vector fields}}.

\subsection{$k$-symplectic Lagrangian field theory}

Let $\cL\in\Cinfty(\Tan^1_kQ)$ be a \emph{Lagrangian function}.

\begin{definition} 
The \textbf{Cartan 1-forms} $\theta^\alpha_\cL\in\df^1(\Tan^1_kQ)$ are defined by
$\theta^\alpha_\cL\coloneqq(J^\alpha)^*\d\Lag$.

The \textbf{Cartan 2-forms} $\omega^\alpha_\cL\in\df^2(\Tan^1_kQ)$ are defined by
$\omega^\alpha_\cL\coloneqq-\d\theta^\alpha_\cL$.

The \textbf{energy Lagrangian function} is
$\cE_\cL\coloneqq\Delta(\cL)-\cL\in\Cinfty(\Tan^1_kQ)$.
\end{definition}

In local coordinates
\bea
\label{PC-local}
\theta^\alpha_\cL&=&\frac{\partial\cL}{\partial\v^i_\alpha}d\x^i  \quad , \quad
\omega^\alpha_\cL \ =\ \frac{\partial^2\cL}{\partial\x^j\partial\v^i_\alpha}d\x^i\wedge d\x^j+\sum_{\beta=1}^{k}\frac{\partial^2\cL}{\partial\v^j_\beta\partial\v^i_\alpha}d\x^i\wedge d\v^j_\beta \ , \\
\nonumber \cE_\cL &=&
\sum_{\alpha=1}^{k} v^i_\alpha\derpar{\cL}{v^i_\alpha}-\Lag(q^j,v^j_\beta) \ .
\eea

\begin{prop}
(\cite{dLSV}).
The following statements are equivalent:
\ben
\item
$(\omega_\cL^1,\ldots,\omega_\cL^k;\Ker\pi_*)$ is a $k$-symplectic structure in $\Tan_k^1Q$.
\item
At every point ${\bf v}_q\in\Tan^1_kQ$, there exists a chart with coordinates $(\x^i;\v_\alpha^i)$, for $1\leq i\leq n$, such that the Hessian matrix
$\ds\Big(\frac{\partial^2\cL}{\partial{\v}^i_\alpha\partial{\v}_\beta^j}({\bf v}_q)\Big)$
is regular.
\een
\end{prop}

\begin{definition}
A Lagrangian function is \textbf{regular} if
the above equivalent conditions hold.
Then $(\Tan^1_kQ;\cL)$ is a \textbf{$k$-symplectic Lagrangian system}.
Otherwise, $\cL$ is a \textbf{singular Lagrangian} and $(\Tan^1_kQ;\cL)$ is a \textbf{$k$-presymplectic Lagrangian system}.
\end{definition}

The Lagrangian field equations are established as follows:

\begin{definition}
Let $\bX=(X_1,\ldots,X_k)\in\vf^k(\Tan^1_kQ)$, the \textbf{$k$-symplectic Lagrangian equation} is
\begin{equation}\label{k-sym_LagEq}
\sum_{\alpha=1}^{k}\inn({X_\alpha})\omega^\alpha_\cL=\d\cE_\cL \ ,
\end{equation}
and $\bX$ is called a \textbf{Lagrangian $k$-vector field}.
In addition, if $\bX$ is a \textsc{sopde} then
\eqref{k-sym_LagEq} is the \textbf{$k$-symplectic Euler--Lagrange equation}
and, if $\bX$ is integrable, 
it is called an \textbf{Euler--Lagrange $k$-vector field}.
\end{definition}

A $k$-vector field $\bX=(X_1,\ldots,X_k)\in\vf^k(T^1_kQ)$ is locally given by
$$
X_\alpha=(X_\alpha)^i\frac{\partial}{\partial\x^{i}}+\sum_{\beta=1}^{k}(X_\alpha)^i_\beta\frac{\partial}{\partial\v^{i}_\beta}\ .
$$ 
Using the local expression of $\bX$, $\omega_\Lag$, and $\cE_\cL$, as
$$
\d\cE_\cL=\bigg(\sum_{\alpha=1}^k\frac{\partial^2\cL}{\partial\x^j\partial\v^i_\alpha}\v^i_\alpha-\frac{\partial\cL}{\partial\x^j}\bigg)d\x^j+\sum_{\alpha,\beta=1}^k\frac{\partial^2\cL}{\partial\v^j_\beta\partial\v^i_\alpha}\v^i_\alpha d\v^j_\beta
$$
we find that the $k$-symplectic Lagrangian equation reads
\beann
\sum_{\alpha=1}^k\Bigg(\frac{\partial^2 {\cal L}}{\partial \x^i\partial \v^j_\alpha}\v_\alpha^j-\bigg[
\frac{\partial^2 {\cal L}}{\partial \x^i \partial \v^j_\alpha} -
\frac{\partial^2 {\cal L}}{\partial \x^j \partial \v^i_\alpha}
\bigg](X_\alpha)^j+
\sum_{\beta=1}^k\frac{\partial^2 {\cal L}}{\partial \v_\alpha^i \partial \v^j_\beta}\,
(X_\alpha)^j_\beta\Bigg)
&=& \frac{\partial {\cal L}}{\partial \x^i} \ ,
\\
\sum_{\alpha=1}^k\frac{\partial^2 {\cal L}}{\partial \v^i_\beta \partial \v^j_\alpha}\,\Big[(X_\alpha)^j-\v_\alpha^j\Big]
&=& 0 \ .
\eeann
If $\cL$ is a regular Lagrangian or $\bX$ is required to be a \textsc{sopde},
then the above equations are equivalent to
\beann
\sum_{\beta=1}^k\Bigg(\frac{\partial^2\cL}{\partial\x^i\partial\v^j_\beta}\v^j_\beta + \sum_{\alpha=1}^k\frac{\partial^2\cL}{\partial\v^i_\beta\partial\v^j_\alpha}(X_\beta)^j_\alpha\Bigg) &=&\frac{\partial\cL}{\partial\x^i}\ , \\
(X_\alpha)^i &=& \v^i_\alpha\ .
\eeann

\begin{theo} \cite{dLSV}
Let $\cL\in\cF(\Tan^1_kQ)$. Then:
\begin{enumerate}[(i)]
\item If $\cL$ is regular, then there exist ${\bX}\in\vf^k(\Tan^1_kQ)$ 
solutions to to the Lagrangian equations (\ref{k-sym_LagEq}) and they are \textsc{sopde}. 
\item If ${\bX}\in\vf^k(\Tan^1_kQ)$ is an Euler--Lagrange vector field and $\tilde{\psi}$ is an integral section of ${\bX}$, then $\psi$ is a solution to the {\rm Euler--Lagrange field equations}
\beann
\sum_{\beta=1}^k\frac{\partial}{\partial t^\beta}\bigg(\frac{\partial\cL}{\partial\v^i_\beta}\circ\psi\bigg) & =&\frac{\partial\cL}{\partial\x^i}\circ\psi\ , \\
\v^i_\alpha\circ\psi & =& \frac{\partial(\x^i\circ\psi)}{\partial t^\alpha},
\eeann
\end{enumerate}
\end{theo}

Thus, if $\cL$ is regular, the existence of solutions to the Lagrangian equations (\ref{k-sym_LagEq}) is assured, although they are neither unique, nor necessarily integrable. Furthermore, they are \textsc{sopde}s and, if $\bX$ is integrable, their integral sections are holonomic and they are solutions to the Euler--Lagrange equations.
If $\cL$ is a singular Lagrangian, then the \sopde condition $(X_\alpha)^i=\v_\alpha^i$ is not obtained from the Lagrangian equations and it must be imposed as an additional condition in order to obtain holonomic solutions to the field equations (as it is required by the variational principles).

As it is usual in mechanics, we will restrict our study 
to a particular family of singular Lagrangians.
First the \emph{\textbf{Legendre transformation}} induced by $\cL$ 
is the map $\Leg\colon\Tan^1_kQ\longrightarrow (\Tan^1_k)^*Q$ given by
$$
\Leg(q;v_{q1},\ldots,v_{qk})\coloneqq\frac{d}{ds}\Big|_{s=0}\cL(q;v_{q1},\ldots,v_{q\alpha}+su_q,\ldots,v_{qk}).
$$
Observe that $\tau\circ\Leg=\pi$ and then
this is a fiber-preserving map.
In natural coordinates we have
$$
\Leg^*q^i=q^i \quad , \quad \Leg^*p_i^\alpha=\derpar{\Lag}{v^i_\alpha} \ .
$$

\begin{definition} A singular Lagrangian function $\cL\in C^\infty(\Tan^1_kQ)$ 
is said to be \textbf{almost-regular} if:
\begin{enumerate}
	\item $\Leg(\Tan^1_kQ)$ is a closed submanifold of $(\Tan^1_k)^*Q$.
	\item $\Leg$ is a submersion onto its image.
	\item The fibers $(\Leg)^{-1}({\bf p}_q)$ for all ${\bf p}_q\in\Leg(\Tan^1_kQ)$ 
	are connected submanifolds of $\Tan^1_kQ$.
\end{enumerate}
\end{definition}

This assumption assures that certain regularity conditions hold which,
in particular, guarantee the existence of a Hamiltonian counterpart for these kinds of systems.
In fact; let $P_0=\Leg(\Tan^1_kQ)$ with the natural embedding
$j_0\colon P_0\hookrightarrow(\Tan^1_k)^*Q$,
and $\Leg_0\colon\Tan^1_KQ\to P_0$ such that
$j_0\colon\Leg_0=\Leg$.
The restriction of the canonical $k$-symplectic structure $(\omega^1,\ldots,\omega^k,\d t)$ of $(\Tan^1_k)^*Q$
to $P_0$ defines a $k$-presymplectic structure 
$(\omega_0^1,\ldots,\omega_0^k;\d t)$ in $P_0$, where $\omega_0^\alpha=j_o^*\omega^\alpha$,
and $\omega_\Lag^\alpha=\Leg_0^*\omega_0^\alpha=\Leg^*\omega^\alpha$. 
Furthermore, the Lagrangian energy $\cE_\cL$ is $\Leg$-projectable, 
then there exists a function $H_0\in\Cinfty(P_0)$ such that
$\Leg_0^*H_0=\cE_\cL$, and
$(P_0,H_0)$ is a {\sl $k$-presymplectic Hamiltonian system} with equations
\beq
\sum_{\alpha=1}^{k}\inn(X^0_\alpha)\omega^\alpha_0=\d H_0 \quad , \quad
{\bf X}^0=(X^0_\alpha)\in\vf^k(P_0) \ .
\label{hameqs}
\eeq
which is $\Leg$-related with the Lagrangian system $(\Tan^1_kQ,\Lag)$
(see \cite{GMR} for the details).


\section{Constraint algorithm for $k$-presymplectic Lagrangian systems}

Next we give a generalisation to $k$-presymplectic Lagrangian systems of the constraint algorithm
presented in \cite{MR} for presymplectic Lagrangian systems.
This new algorithm generalizes what was developed in \cite{GMR},
incorporating the \sopde condition for the solutions.

\subsection{Statement of the problem and previous considerations}

Let $(\Tan^1_kQ;\cL)$ be a $k$-presymplectic Lagrangian system,
where $\Lag\in\Cinfty(\Tan^1_kQ)$ is an almost-regular Lagrangian.

The problem we want to solve consists in finding a submanifold 
$S$ of $\Tan^1_kQ$ and a $k$-vector field $\bX=(X_1,\ldots,X_k)\in\vf^k(\Tan^1_kQ)$ 
such that, on the points of $S$, 
\begin{enumerate}
	\item $\bX$ is a solution to the Lagrangian equation
	$\ds\sum_{\alpha=1}^k\inn(X_\alpha)\omega^\alpha_\cL=\d\cE_\cL$,
	\item $\bX$ is a \textsc{sopde}, and
	\item $\bX$ is tangent to $S$.
\end{enumerate}
In addition, $\bX$ should be an integrable $k$-vector field in $C$. Nevertheless, the integrability problem will not be considered in depth in this work.

\noindent{\bf Notation:} The set of $k$-vector fields on $\Tan^1_kQ$ that are solutions
to the Lagrangian equation on a subset $S\subseteq \Tan^1_kQ$ 
is denoted $\fX^k_\cL(\Tan^1_k Q)|_S$. The set of vector fields on $\Tan^1_kQ$ that are 
\textsc{sopde}s on a subset $S\subseteq\Tan^1_kQ$ is $\fX^k_\cS(\Tan^1_kQ)|_S$. 
We write 
$\fX^k_{\cL,\cS}(\Tan^1_kQ)_S\coloneqq\fX^k_{\cL}(\Tan^1_kQ)|_S\cap\fX^k_{\cS}(\Tan^1_kQ)|_S$. 
The set of vertical vector fields in $\Tan^1_kQ$ for the projection 
$\pi\colon\Tan^1_kQ\to Q$ is denoted by $\fX^V(\Tan^1_kQ)$ and,
similarly, $\fX^{kV}(\Tan^1_kQ)$ denotes the set of vertical $k$-vector fields for $\pi$, 
i.e. $\bX=(X_\alpha)$ such that $X_\alpha\in\fX^V(\Tan^1_kQ)$.

In the following, we consider the map
\begin{equation*}
\begin{split}
\bOm_\cL\ :\ \ \ \ \ \ \ \ \ \ \ \ \fX^k(\Tan^1_kQ)& \longrightarrow\ \bOm(\Tan^1_kQ)\\
\bX=(X_1,\ldots,X_k)& \longmapsto\ \bOm_\cL(\bX)\coloneqq\sum_{\alpha=1}^k\inn(X_\alpha)\omega^\alpha_\cL.
\end{split}
\end{equation*}
Then, for ${\mathfrak S}\subseteq\fX^k(\Tan^1_kQ)$ and $\bOm_\cL$,
the {\sl annihilator} of ${\mathfrak S}$ is defined as
$$
{\mathfrak S}^\perp\coloneqq\{X\in\fX(\Tan^1_kQ)\ |\ \bOm_\Lag(\bY)(X)=0,\ \mbox{\rm for every $\bY\in {\mathfrak S}$}\} \ .
$$ 

The following properties hold:

\begin{lemma}
\label{lemuno}
$\ds [\fX^k(\Tan^1_kQ)]^\perp=\bigcap_{\alpha=1}^k\Ker\omega^\alpha_\Lag$.
\end{lemma}
\proof\
For every $(Y_1,\ldots,Y_k)\in\vf^k(\Tan^1_kQ)$,
if $\ds X\in \bigcap_{\alpha=1}^k \Ker\omega_\Lag^\alpha$ \ we have
$$
\Big(\sum_{\alpha=1}^k\inn(Y_\alpha)\omega^\alpha_\Lag\Big)(X) =
-\sum_{\alpha=1}^k\inn(Y_\alpha)\inn(X)\omega^\alpha_\Lag=0
\ \Longrightarrow \ X\in[\fX^k(\Tan^1_kQ)]^\perp \ .
$$
Conversely, if $X\in[\fX^k(\Tan^1_kQ)]^\perp$,
then $\ds\sum_{\alpha=1}^k\inn(Y_\alpha)\inn(X)\omega^\alpha_\Lag=0$,
for every $(Y_1,\ldots,Y_k)\in\vf^k(\Tan^1_kQ)$.
Then, taking any $(Y_1,0,\ldots,0)$ with $Y_1\neq 0$,
we conclude that $X\in\Ker\omega^1_\Lag$;
and analogously for the others.
\qed

\begin{lemma}
\label{lemdos}
$\ds [\vf^k(\Tan^1_kQ)]^\perp\cap\vf^V(\Tan^1_kQ)=\Ker\Leg_*$.
\end{lemma}
\begin{proof} 
Remember that, by Lemma \ref{lemuno}, $\ds [\vf^k(\Tan^1_kQ)]^\perp=\bigcap_{\alpha=1}^k\Ker\omega_\Lag^\alpha$.
First, as $\tau\circ\Leg=\pi$ (that is, $\Leg$ is fiber-preserving), then
$\Ker\Leg_*\subset\vf^V(\Tan^1_kQ)$. Furthermore, for every $Z\in\Ker\Leg_*$,
\beq
\inn(Z)\omega^\alpha_\Lag=\inn(Z)(\Leg_0^*\omega_0^\alpha)=\Leg_0^*[\inn(\Leg_{0*}Z)\omega_0^\alpha]=0 \ ;
\label{keromegas}
\eeq
since $\Leg_{0*}Z=0$; hence $Z\in\Ker\omega_\Lag^\alpha$,
for every $\alpha$; therefore $\ds Z\in\bigcap_{\alpha=1}^k\Ker\omega_\Lag^\alpha$,
and then we conclude that 
$$
Z\in\Big(\bigcap_{\alpha=1}^k\Ker\omega_\Lag^\alpha\Big)\cap\vf^V(\Tan^1_kQ)=[\vf^k(\Tan^1_kQ)]^\perp\cap\vf^V(\Tan^1_kQ) \ .
$$
Conversely, for every $\alpha$,
from the coordinate expression \eqref{PC-local} we see that a vector field 
$\dst Z=\sum_{\alpha=1}^kX^i_\alpha\frac{\partial}{\partial\v^i_\alpha}\in\vf^V(\Tan^1_kQ)$
belongs to $\Ker\omega_\Lag^\alpha$ if, and only if, it satisfies that
$\ds\sum_{\alpha=1}^kX^i_\alpha\frac{\partial^2\cL}{\partial\v^i_\alpha\v^j_\beta}=0$.
Then, taking into account that the matrix representing $\Leg_*$ is
$\ds
\left(\begin{matrix}
\left({\rm Id}\right)_{n\times n}&\left( 0 \right)_{n\times nk}\\
\left(\ds\frac{\partial^2\Lag}{\partial q^i\partial v^j_\alpha}\right)&\left(\ds\frac{\partial^2\Lag}{\partial v^j_\beta\partial v^i_\alpha}\right)
\end{matrix}
\right)
$,
a simple calculation leads to $Z\in\Ker\Leg_*$.
\end{proof}


\subsection{Compatibility conditions: First generation constraints}

We start by imposing the compatibility condition. The subset of $\Tan^1_kQ$ where the Lagrangian equation \eqref{k-sym_LagEq} has a solution is 
$$
P_1\coloneqq\big\{x\in\Tan^1_kQ\ \big|\ (\d\cE_{\Lag})_x\in{\rm Im}\,(\bOm_{\Lag})_x\subseteq \Tan^*_x(\Tan^1_kQ)\big\}
$$
and it is assumed to be a closed submanifold of $\Tan^1_kQ$.
Then, using Lemma \ref{lemuno} one can prove that \cite{GMR}:

\begin{prop}
$P_1=\big\{x\in\Tan^1_kQ\ \big|\ (\inn(Z)\d\cE_\Lag)(x)=0,\ \mbox{\rm for every  $Z\in[\fX^k(\Tan^1_kQ)]^\perp$}\big\}.$

Moreover, if \ $\bX\in\fX^k_\cL(T^1_kQ)|_{P_1}$, then \
$\fX^k_\cL(\Tan^1_kQ)|_{P_1}=\big\{\bX+\bY\ \big|\ \bY\in\Ker\bOm_\Lag|_{P_1}\big\}$.
\end{prop}

\begin{definition}
The functions
$$
\zeta^{(Z)}_1\coloneqq\inn(Z)\d\cE_\Lag \ , \ \mbox{\rm for $Z\in[\fX^k(T^1_kQ)]^\perp$} \ ,
$$
are called \textbf{first generation $k$-presymplectic} or \textbf{dynamical constraints}.
They define the submanifold $P_1$.
\end{definition}

The name of these constraints refers to the fact that 
their origin has nothing to do with the second-order condition, 
but only with the compatibility of Lagrangian equations in general.
As it happens in mechanics \cite{BGPR,CLR,got79,MR},
the following property characterizes these kinds of constraints:

\begin{prop}
First generation  dynamical constraints can be expressed 
as $\Leg$-projectable functions.
\label{dynpro}
\end{prop}
\begin{proof}
First notice that, as $\Leg$ is a submersion (and hence $\Leg_0$ is too),
for every
$Z_0\in\cap_\alpha\Ker\omega_0^\alpha\subset\vf((\Tan^1_k)^*Q)$ 
there exist $Z\in\vf(\Tan^1_kQ)$ such that $\Leg_{0*}Z=Z_0$.
In addition, from \eqref{keromegas} we have that
$\inn(Z)\omega_\Lag^\alpha=\Leg_0^*[\inn(Z_0)\omega_0^\alpha]=0$,
for every $\alpha$; therefore $Z\in\cap_\alpha\Ker\omega_\Lag^\alpha$
and, furthermore, if $Z,Z'\in\cap_\alpha\Ker\omega_\Lag^\alpha$
are such that $\Leg_{0*}Z=\Leg_{0*}Z'=Z_0$, then $Z-Z'\in\Ker\Leg_*$.
As $Z\in\cap_\alpha\Ker\omega_\Lag^\alpha$ is $\Leg$-projectable
then $[V,Z]\in\Ker\Leg_*$, for every $V\in\Ker\Leg_*$. 

The necessary and sufficient condition for a function $\zeta^{(Z)}_1\in\Cinfty(\Tan^1_kQ)$
to be $\Leg$-projectable is that $V(\zeta^{(Z)}_1)=0$, 
for every $V\in\Ker\Leg_*$. Therefore, taking a local set of generators
of $\ds\bigcap_{\alpha=1}^k\Ker\omega_\Lag^\alpha$ made of these 
$\Leg$-projectable vector fields, for the corresponding base of dynamical constraints $\zeta^{(Z)}_1$ we obtain
$$
\Lie(V)(\zeta^{(Z)}_1)=\Lie(V)\inn(Z)\d\cE_\Lag=\Lie(V)\Lie(Z)\cE_\Lag=
\Lie([V,Z])\cE_\Lag-\Lie(Z)\Lie(V)\cE_\Lag=0 \ ,
$$
since $\Lie(V)\cE_\Lag=0$, because $\cE_\Lag$ is $\Leg$-projectable too.
\end{proof}

So, we have found a submanifold $P_1\hookrightarrow\Tan^1_kQ$
here we can find solutions to the Lagrangian equation \eqref{k-sym_LagEq},
but they are not \textsc{sopde}s necessarily.
Our aim now is to find the largest subset of $P_1$ such that some of the solutions in $\fX^k_\cL(\Tan^1_kQ)|_{P_1}$ can be chosen to be \textsc{sopde}s. 
Then, if $\bX\in\fX^k_\cL(T^1_kQ)|_{P_1}$ is a solution on $P_1$, we define
\beq
S_1\coloneqq\big\{x\in P_1\ \big|\ \exists\bY\in\Ker\bOm_\Lag,\ \sum_{\alpha=1}^kJ^\alpha(X_\alpha+Y_\alpha)_x=\Delta_x\big\}.
\label{modsopde}
\eeq
which is assumed to be a closed submanifold of $P_1$.
It is clear by its definition that $S_1$ is the maximal subset of $P_1$ 
where we can find solutions to the Lagrangian equation \eqref{k-sym_LagEq}
that are also \textsc{sopde}s.

Now, let $\fM\coloneqq[\fX^{kV}(\Tan^1_kQ)]^\perp$.

\begin{theo}\label{kS_1} 
Let $\bX\in\fX^k_\cL(\Tan^1_kQ)|_{P_1}$ be a Lagrangian $k$-vector field and 
$\bY\in\fX^k(\Tan^1_kQ)$ be any $k$-vector field 
such that $\bX+\bY$ is a \textsc{sopde}. Then
$$
S_1=\big\{x\in P_1\ \big|\ \big(\inn(Z)\bOm_\Lag(\bY)\big)(x)=0,\ 
\mbox{\rm for every $Z\in\fM$}\big\} \ .
$$
\end{theo}

In order to prove this theorem, we need the following lemma:

\begin{lemma}\label{k-exten} If $\bY\in\fM^\perp$, then there exists $\bV\in\fX^{kV}(\Tan^1_kQ)$ such that $\bY-\bV\in\Ker\bOm_\Lag$.
\end{lemma}
\begin{proof} 
The equation $\bOm_\Lag(\bY)=\bOm_\Lag(\bV)$ can be solved for $\bV\in\fX^{kV}(\Tan^1_kQ)$
if, and only if, for every $U\in[\fX^{kV}(\Tan^1_kQ)]^\perp=\fM$, we have that
$\inn(U)\bOm_\Lag(\bY)=0$; and this holds since $\bY\in\fM^\perp=\big[[\fX^{kV}(\Tan^1_kQ)]^\perp\big]^\perp=\fX^{kV}(\Tan^1_kQ)$.
\end{proof}

\noindent\textit{Proof of Theorem \ref{kS_1}.} 
Let $\mathcal{S}\coloneqq\big\{x\in P_1\ \big|\ \big(\inn(Z)\bOm_\Lag(\bY)\big)(x)=0,\ \mbox{\rm for every  $Z\in\fM$}\big\}$. 
It is clear that this definition does not depend on the choice of solution $\bX$ 
since any two solutions differ by an element of $\Ker\bOm_\Lag$. 
One can indeed prove that it is also independent of the choice of $\bY$:
in fact, let $\tilde\bY$ be another $k$-vector field such that $\bX+\tilde\bY$ is a \textsc{sopde}; 
then $\bY-\tilde\bY$ is a $\pi$-vertical $k$-vector field 
since it is the difference of two \textsc{sopde}s,
and then, for every $Z\in\fM$,
$$
\inn(Z)\bOm_\Lag(\bY-\tilde\bY)=0 \ .
$$

Now, if $x\in S_1$, then there exists $\bY\in\Ker\bOm_\Lag$ 
such that $\bX+\bY$ is a \textsc{sopde} and is a solution 
to the Lagrangian equation \eqref{k-sym_LagEq} at $x$. 
Hence $\bOm_\Lag(\bY)_x=0$ and, in particular, $(\inn(Z)\bOm_\Lag(\bY))(x)=0$,
for every $Z\in\fM$. So, $x\in\mathcal{S}$ and $S_1\subseteq\mathcal{S}$.

Conversely, if $x\in \mathcal{S}$, then $(\inn(Z)\bOm_\Lag(\bY))(x)=0$,
for every $Z\in\fM$. Hence, $\bY_x\in\fM^\perp_x$ and,
by Lemma \ref{k-exten}, there exists a $\pi$-vertical $k$-vector 
$\mathbf{v}\in (\Tan^1_k)_x(\Tan^1_kQ)$ such that $\bY_x+\mathbf{v}\in(\Ker\bOm_\Lag)_x$. 
Taking any $\bV\in\fX^k(\Tan^1_kQ)$ such that $\bV_x=\mathbf{v}$ we get that 
$\bX+\bY+\bV$ solves the Lagrangian equation \eqref{k-sym_LagEq} 
and is a \textsc{sopde} at $x$. So $x\in S_1$ and $\mathcal{S}\subseteq S_1$.
\qed

\begin{prop} 
\label{solsopd}
If $\bD\in\fX^k_{\cL,\cS}(\Tan^1_kQ)|_{S_1}$, then
$$
\fX^k_{\cL,\cS}(\Tan^1_kQ)|_{S_1}=\big\{\bD+\bV\in\fX^k(\Tan^1_kQ)\ \big|\ \bV\in\Ker^V\bOm_\Lag|_{S_1}\big\} \ ,
$$
where $\Ker^V\bOm_\Lag\equiv\Ker\bOm_\Lag\cap\fX^{kV}(T^1_kQ)$.
\end{prop}
\begin{proof} 
It is clear that adding an element of $\Ker^V\bOm_\Lag$ 
to a \sopde solution to the Lagrangian equation \eqref{k-sym_LagEq}
we obtain another \sopde solution. Conversely, if $\bD,\tilde\bD\in\fX^k_{\cL,\cS}(\Tan^1_kQ)|_{S_1}$ we have
$$
\bOm_\Lag(\bD-\tilde\bD)=0 \quad , \quad J^\alpha(\Gamma_\alpha-\tilde \Gamma_\alpha)=0,
$$
so $\bX-\tilde\bX$ is vertical and belongs to $\Ker\bOm_\Lag$.
\end{proof}

We have the following equivalent characterization of the submanifold $S_1$:

\begin{prop}\label{uni_k} 
For every \textsc{sopde} $\bD$ we have
$$
S_1=\big\{x\in\Tan^1_kQ\ \big|\ \inn(Z)\big(\bOm_\Lag(\bD)-\d\cE_\Lag\big)(x)=0,\ 
\mbox{\rm for every $Z\in\fM$} \big\}.
$$
\end{prop}
\begin{proof} 
First, observe that 
$[\fX^k(\Tan^1_kQ)]^\perp\subseteq[\fX^{kV}(\Tan^1_kQ)]^\perp=\fM$.
Then, for $Z\in[\fX^k(\Tan^1_kQ)]^\perp\subseteq\fM$, we have that
$$
0=\inn(Z)\big(\bOm_\Lag(\bD)-\d\cE_\Lag\big)(x)=-\inn(Z)\d\cE_\Lag(x)\ ;
$$
that is, we recover the dynamical constraints. 

On the other hand, if $Z\in\fM\smallsetminus[\fX^k(\Tan^1_kQ)]^\perp$, writing $\bD=\bX+\bY$, with $\bX\in\fX^k_\cL(T^1_kQ)$, we have
$$
0=\inn(Z)\big(\bOm_\Lag(\bD)-\d\cE_\Lag\big)(x)=
\big(\inn(Z)\bOm_\Lag(\bY)\big)(x)+\inn(Z)\big(\bOm_\Lag(\bX)-\d\cE_\Lag\big)(x)=
\big(\inn(Z)\bOm_\Lag(\bY)\big)(x) \ .
$$
In the last equality we have used that the second term vanishes in the points where the dynamical constraints vanish, i.e. $P_1$; 
and we are precisely looking for the set where all the constraints vanish.
\end{proof}

\begin{definition}
For every $\bX\in\fX^k_\cL(T^1_kQ)|_{P_1}$, any \textsc{sopde} $\bD$, and $Z\in\fM\smallsetminus[\fX^k(T^1_kQ)]^\perp$, the functions 
\beq
\eta^{(Z)}_1\coloneqq\inn(Z)\bOm_\Lag(\bD-\bX)
\label{nondyncons}
\eeq
are called \textbf{first generation \sopde} or \textbf{non-dynamical constraints}.
They define $S_1\hookrightarrow P_1$.
\end{definition}

This name is justified because these constraints appear as a consequence of demanding the \sopde condition.
They are characterized by the following condition:

\begin{prop}
\label{nodynpro}
First generation \sopde constraints are not $\Leg$-projectable.
\end{prop}
\begin{proof}
As $\Leg$ is a submersion, following the same reasoning than
at the begining of the proof of Proposition \ref{dynpro},
first we prove that, from every ${\bf X}_0\in\vf^k(P_0)$
solution to the Hamiltonian equations \eqref{hameqs},
we construct a $\Leg$-projectable solution ${\bf X}\in\vf^k(\Tan^1_kQ)$
to the Lagrangian equations \eqref{k-sym_LagEq}.
In the same way, we can obtain a set of vector fields
$Z\in\fM\smallsetminus[\fX^k(T^1_kQ)]^\perp$ which generate (locally) 
$\fM\smallsetminus[\fX^k(T^1_kQ)]^\perp$ and are $\Leg$-projectable too.
Therefore, taking this solution $\bX$,
the part $\inn(Z)\bOm_\Lag(\bX)$ in the expression \eqref{nondyncons} gives a $\Leg$-projectable function.
However, for $\inn(Z)\bOm_\Lag(\bD)$ we have that,
for every $V\in\Ker\Leg_*$,
\beann
\Lie(V)\inn(Z)\bOm_\Lag(\bD)&=& \Lie(V)\inn(Z)\sum_\alpha\inn(\Gamma_\alpha)\omega_\Lag^\alpha =
-\Lie(V)\sum_\alpha\inn(\Gamma_\alpha)\inn(Z)\omega_\Lag^\alpha
\\
&=& -\sum_\alpha\inn([V,\Gamma_\alpha])\inn(Z)\omega_\Lag^\alpha
+\sum_\alpha\inn(\Gamma_\alpha)\Lie(V)\inn(Z)\omega_\Lag^\alpha=
-\sum_\alpha\inn([V,\Gamma_\alpha])\inn(Z)\omega_\Lag^\alpha \ ,
\eeann
since $\Lie(V)\inn(Z)\omega_\Lag^\alpha=0$, because $\inn(Z)\omega_\Lag^\alpha$ is a $\Leg$-projectable
$1$-form.
But $Z\not\in[\fX^k(T^1_kQ)]^\perp=\cap_\alpha\Ker\omega_\Lag^\alpha$, by the hypothesis, and $[V,\Gamma_\alpha]\not\in\Ker\omega_\Lag$, for every $\alpha$,
since $\bD$ is an arbitrary \sopde. Therefore
$\Lie(V)\inn(Z)\bOm_\Lag(\bD)\not=0$, which implies that
$\inn(Z)\bOm_\Lag(\bD)$ and hence 
$\eta^{(Z)}_1=\inn(Z)\bOm_\Lag(\bD-\bX)$
are not $\Leg$-projectable.
\end{proof}

This means that these constraints remove degrees of freedom
on the fibers of the foliation defined
in $\Tan^1_kQ$ by the Legendre map $\Leg$
and, as a consequence, $\Leg(P_1)=\Leg(S_1)$:

\subsection{Tangency conditions: Second and further generation constraints}

In general, none of the elements of $\fX^k_{\cL,\cS}(\Tan^1_kQ)|_{S_1}$ 
may be tangent to $S_1$. 
Thus, bearing in mind Proposition \ref{solsopd},
we have to look for the subset of $S_1$ where,
given any solution $\bD\in\fX^k_{\cL,\cS}(\Tan^1_kQ)|_{S_1}$, 
we can find an element $\bV$ of $\Ker^V\bOm_\Lag$
such that it renders $\bD+\bV$ tangent to $S_1$.

At this point, the situation for $k$-presymplectic Lagrangian systems 
differs from the one in Lagrangian mechanics ($k=1$). 
The main difference is that now, to every Euler--Lagrange $k$-vector field 
we can add elements of $\Ker^V\bOm_\Lag$; but
$(\Ker\omega_\Lag^1,\ldots,\Ker\omega_\Lag^k)\subset\Ker\bOm_\Lag$,
and this means that, if $\bV=(V_\alpha)$, 
then $V_\alpha\not\in\Ker\omega_\Lag^\alpha$ necessarily.
Despite this, the origin and structure of new generation of constraints
is similar in mechanics and in field theories, as we will see below.

The submanifold $S_1$ is defined as the zero set of constraint functions, 
and hence we can easily impose tangency conditions.
So, for $\bD\in\fX^k_{\cL,\cS}(T^1_kQ)|_{S_1}$, 
bearing in mind Proposition \ref{solsopd}, we define
\begin{equation*}
\begin{split}
S_2\coloneqq\Big\{x\in S_1\ \Big| \ \exists\bV\in\Ker^V\bOm_\Lag \ \vert\  &
 (\Gamma_\alpha+V_\alpha)(\zeta^{(Z)}_1)(x)=0,\ \mbox{for every $Z\in[\fX^k(T^1_kQ)]^\perp$}\ , \\
 &(\Gamma_\alpha+V_\alpha)(\eta^{(Z)}_1)(x)=0,\ \mbox{for every $Z\in\fM\smallsetminus[\fX^k(T^1_kQ)]^\perp$}\ \Big\}\ .
\end{split}
\end{equation*}
and we assume that $S_2$ is a closed submanifold of $S_1$.

\begin{lemma} For a given $\bD\in\fX^k_{\cL,\cS}(T^1_kQ)|_{S_1}$, 
the system of equations $\Gamma_\alpha(\eta^{(Z)}_1)=-V_\alpha(\eta^{(Z)}_1)$ 
has a solution with $V_\alpha\in\Ker^V\omega_\Lag^\alpha$, on $S_1$
\end{lemma}
\begin{proof} 
We can take a finite set of non-dynamical constraints
$\{\eta^{(Z_i)}_1\}$.
As it has been pointed out, these constraints
remove degrees of freedom
on the leaves of the distribution generated by $\Ker\Leg_*$;
that is, in the vertical leaves of $\Tan_k^1Q$.
Then, in a chart of natural coordinates, the matrix of the linear system
$\Gamma_\alpha(\eta^{(Z_i)}_1)=-V_\alpha(\eta^{(Z_i)}_1)$
(which consists in partial derivatives of the independent constraints
with respect to the coordinates in the vertical fibres)
has maximal rank and, hence,
the system is compatible, at least locally.
Then, from the local solutions we can construct global solutions
using partitions of unity.
(For the case of mechanics, see also \cite{BGPR,CLR,CLR2}).
\end{proof}

From this last result we obtain that, if new constraints appear,
out of all the conditions that define $S_2\hookrightarrow S_1$
the only new constraints arise form the conditions 
$$
(\Gamma_\alpha+V_\alpha)(\zeta^{(Z)}_1)=0 \ .
$$
Now, recall the construction made in \eqref{modsopde}
in which, to any Lagrangian $k$-vector field $\bX$ on $P_1$,
we added an element of $\Ker\bOm_\Lag$ 
so as to render it a \sopde on $S_1$. 
Thus, we can split $\bD=\bX+\bY+\bV$, 
with $\bY\in\Ker\bOm_\Lag$ and $\bV\in\Ker^V\bOm_\Lag$,
and such that $\bX+\bV$ is a Lagrangian $k$-vector field
(which can be taken to be $\Leg$-projectable).
Then, we have two different situations:
\begin{itemize}
\item[(i)]
If $Y_\alpha(\zeta^{(Z)}_1)\vert_{S_1}=0$, then the condition
$$
(X_\alpha+V_\alpha)(\zeta_1^{(Z)})\vert_{S_1}=0
$$
could determine (partially or totally) $\bV$ and\,/\,or
originate new constraints. In this last case:
\eit

\begin{definition}
The functions
\beq
\zeta_2^{(Z)}\coloneqq (X_\alpha+V_\alpha)(\zeta_1^{(Z)})
\label{2dyncon}
\eeq
are called 
\textbf{second generation $k$-presymplectic} or
\textbf{dynamical constraints}.
\end{definition}

As can be seen, these new constraints would arise from demanding the tangency condition
of the Lagrangian $k$-vector field $\bX+\bV$ on the submanifold $P_1$ and are independent on the \sopde condition.
Thus, together with the first generation dynamical constraints
$\zeta_1^{(Z)}$,
they define a submanifold $P_2\hookrightarrow P_1$
where there exist Lagrangian $k$-vector fields
(solutions to the Lagrangian equations \eqref{k-sym_LagEq},
not necessarily \sopde),
which are tangent to $P_1$ on the points of $P_2$,
but that are generally not tangent to $P_2$.

Furthermore, as for the first generation dynamical constraints, we have:

\begin{prop}
Second generation dynamical constraints can be expressed as $\Leg$-projectable functions.
\end{prop}
\begin{proof}
It is immediate since, in the expression \eqref{2dyncon}, 
the solution $\bX+\bV$ and the constraints $\zeta_1^{(Z)}$ 
can be taken $\Leg$-projectable.
\end{proof}

\bit
\item[(ii)]
If $Y_\alpha(\zeta^{(Z)}_1)\vert_{S_1}\neq 0$, then 
$$
(X_\alpha+Y_\alpha+V_\alpha)(\zeta^{(Z)}_1)\vert_{S_1}=0
$$
could determine (partially or totally) $\bV$ and\,/\,or
originate new constraints that appear as a consequence of the \sopde condition and hence:
\end{itemize} 
\begin{definition}
The functions
$$
\eta^{(Z)}_2\coloneqq(X_\alpha+Y_\alpha+V_\alpha)(\zeta^{(Z)}_1)
$$
are called \textbf{second generation \sopde} or
\textbf{non-dynamical constraints}.
\end{definition}

And, as for the second generation non dynamical constraints we have:

\begin{prop}
Second generation \sopde constraints are not $\Leg$-projectable.
\end{prop}
\begin{proof}
For every $V\in\Ker\Leg_*$, we have
$$
\Lie(V)\eta^{(Z)}_2=
\Lie(V)\Lie(\Gamma_\alpha)\zeta^{(Z)}_1=
-\Lie([V,\Gamma_\alpha])\zeta^{(Z)}_1+\Lie(\Gamma_\alpha)\Lie(V)\zeta^{(Z)}_1 \ .
$$
But, as $\zeta^{(Z)}_1$ can be taken $\Leg$-projectable, then
$\Lie(V)\zeta^{(Z)}_1\not=0$. Furthermore,
a simple inspection of their coordinate expressions show
that \sopde $k$-vector fields are not $\Leg$-projectable,
hence $[V,\Gamma_\alpha]\not\in\Ker\Leg_*$ and, in particular,
$[V,\Gamma_\alpha]\not=0$.
Therefore we conclude that $\Lie(V)\eta^{(Z)}_2\not=0$
and then $\eta^{(Z)}_2$ is not $\Leg$-projectable.
\end{proof}

Thus we have that
$$
S_2\coloneqq\big\{x\in S_1\ \big|\ \eta_2^{(Z)}(x)=0,\ \zeta_2^{(Z)}(x)=0,
\ \mbox{for every $Z\in[\fX^k(T^1_kQ)]^\perp$}\big\} \ ,
$$
and we can find Euler--Lagrange $k$-vector fields that are tangent to $S_1$ on the points of $S_2$,
but that are generally not tangent to $S_2$.

At this point, we are in a situation equivalent to the one before imposing tangency. 
We therefore keep repeating the last step; that is, imposing  tangency,
At every step of the algorithm we repeat the same reasoning
and we obtain similar results than in the above step.
In particular, we find new constraints
which, in general, split into two groups: the {\sl dynamical}
and the {\sl \sopde constraints} which,
as above, are also characterized by the fact of being
$\Leg$-projectable or non $\Leg$-projectable, respectively.
Both of them arise from the tangency condition of the previous
dynamical constraints,
whereas the tangency condition on the previous
\sopde constraints does not give new constraints.

The procedure continues until the algorithm stabilizes, i.e. until $S_{i+1}=S_i\eqqcolon S_f$.
The only interesting case from the physical 
point of view is when $S_f$ is a submanifold of $\Tan^1_kQ$, 
which is called the \emph{\textbf{final constraint submanifold}}.
On it, we can find solutions to the second-order problem for the Lagrangian system $(\Tan^1_kQ,\Lag)$.

\subsection{On the $\Leg$-projectability and integrability of solutions}

A remaining problem concerns to the $\Leg$-projectability of
the \sopde solutions found on the final constraint submanifold $S_f$ of the algorithm.
In fact a necessary condition for the existence of 
$\Leg$-projectable \sopde $k$-vector fields 
which are solutions to the Lagrangian equations \eqref{k-sym_LagEq}
on the final constraint submanifold $S_f$
is that $S_f$ only contain one point in every fibre 
of the foliation defined by $\Ker\Leg_*$
(i.e; the fibres of the Legendre map).
In fact, suposse that this condition does not hold
and there are two points $x_A,x_B\in S_f$ in the same fibre
of the foliation, whose coordinates are $(q^i,v_{A\mu}^i)$
and $(q^i,v_{B\mu}^i)$, respectively.
If $\bD=(\Gamma_\alpha)$ is a \sopde solution, then
$$
\Gamma_\alpha(x_A)=v_{A\alpha}^i\derpar{}{q^i}\Big\vert_{x_A}+
\sum_{\beta=1}^k(X_\alpha)^i_\beta(x_A)\frac{\partial}{\partial \v^i_\beta}\Big\vert_{x_A}
\ , \
\Gamma_\alpha(x_B)=v_{B\alpha}^i\derpar{}{q^i}\Big\vert_{x_B}+
\sum_{\beta=1}^k(X_\alpha)^i_\beta(x_B)\frac{\partial}{\partial \v^i_\beta}\Big\vert_{x_B}\ ,
$$
and $\Leg(x_A)=\Leg(x_B)$, 
but $\Leg_*(\Gamma_\alpha(x_A))\not=\Leg_*(\Gamma_\alpha(x_B))$;
hence $\Gamma_\alpha$ is not $\Leg$-projectable.
In conclusion, the submanifold of $S_f$ where
these $\Leg$-projectable \sopde solutions exist must be
diffeomorphic to the quotient $S_f/\Ker\Leg_*$.
(This problem was studied in detail for the case of mechanics, in \cite{got80}).

As a final remark and in order to solve the problem completely, 
we want the \sopde $k$-vector fields on $S_f$ to be integrable. 
In general, the existence of such $k$-vector fields is not assured 
(even if the Lagrangian is regular). 
Thus, if $\bX\in\fX^k_{\cL,\cS}(\Tan^1_kQ)|_{S_f}$ with $\bX=(X_1,\ldots,X_k)$, 
the necessary and sufficient condition for $\bX$ to be integrable is
$[X_\alpha,X_\beta]\vert_{S_f}=0$.
In most cases, these conditions lead to some relations among 
the remaining arbitrary coefficients of the family of solutions; 
but, in some cases, they can originate new constraints 
that define a new submanifold $S_f'\hookrightarrow S_f$. 
In this case, the tangency of the integrable family of solutions 
has to be checked and the algorithm restarts once again.


\section{The Einstein-Palatini model of General Relativity}

One of the most interesting singular classical field theories in physics
which can be described using the $k$-(pre)symplectic formulation
is the {\sl Einstein-Palatini model} of General Relativity;
also known as the {\sl metric-affine model}.
As we will see, it is described by an affine Lagrangian.
Then, in order to study this model using  the constraint algorithm,
it is very relevant to apply it first to the
generic case of $k$-presymplectic Lagrangian systems described by 
affine Lagrangians.

\subsection{Affine Lagrangians}
\label{subsubsect:GC}

An {\sl affine Lagrangian} is the sum of two functions:
a linear function $T^1_kQ \to \R$ on the fibers of the bundle $\Tan^1_kQ\to Q$ 
and the pullback to $\Tan^1_kQ$ of another arbitrary function in $Q$.
In natural coordinates, $(\x^i,\v^i_\alpha)$, it has the shape
\beq
\cL=\sum_{\alpha=1}^kF^\alpha_i(\x^1,\ldots,\x^k)\v_\alpha^i+G(\x^1,\ldots,\x^k) \ ;
\label{aflag}
\eeq
and then we have
\beann
\Delta(\cL)=\sum_{\alpha=1}^k\v^i_\alpha F^\alpha_i
\quad \Longrightarrow\quad \cE_\cL=-G
\quad \Longrightarrow\quad
\d\cE_\cL=-\frac{\partial G}{\partial\x^i}\d\x^i ,
\\ 
\theta_\cL^\alpha=\frac{\partial\cL}{\partial\v^i_\alpha}d\x^i=F_i^\alpha d\x^i
\quad \Longrightarrow\quad
\omega^\alpha_\cL=-\frac{\partial F^\alpha_i}{\partial\x^j}\ d\x^j\wedge d\x^i \ .
\eeann
For a $k$-vector field $X=(X_\alpha)\in\vf^k(\Tan^1_kQ)$, if
$\ds X_\alpha=(X_\alpha)^i\frac{\partial}{\partial\x^i}+
\sum_{\beta=1}^k(X_\alpha)^i_\beta\frac{\partial}{\partial\v^i_\beta}$,
the Lagrangian equation \eqref{k-sym_LagEq} read
\beq
\label{aff-Leq}
\Bigg[\sum_{\alpha=1}^k(X_\alpha)^i\bigg(\frac{\partial F^\alpha_i}{\partial\x^j}-
\frac{\partial F^\alpha_j}{\partial\x^i}\bigg)+
\frac{\partial G}{\partial\x^j}\Bigg]\,\d\x^j=0 
\ \Longleftrightarrow\
\sum_{\alpha=1}^k(X_\alpha)^i\bigg(\frac{\partial F^\alpha_i}{\partial\x^j}-
\frac{\partial F^\alpha_j}{\partial\x^i}\bigg)+
\frac{\partial G}{\partial\x^j}=0 \ .
\eeq
The behaviour of the system depends on the rank of the matrix 
$\ds{\cal M}=(M^\alpha_{ij})=\left(\frac{\partial F^\alpha_i}{\partial\x^j}-\frac{\partial F^\alpha_j}{\partial\x^i}\right)$,
which is assumed to be constant (otherwise the analysis must be done 
on each one of the subsets of $\Tan^1_kQ$ where it is constant).
Now, we follow the steps of the algorithm.

\noindent\textbf{First-generation dynamical constraints:}
The algorithm tells us that
$$
P_1=\big\{x\in\Tan^1_kQ\ \big|\ (\inn(Z)\d\cE_\cL)(x)=0,\ \mbox{\rm for every $Z\in[\fX^k(\Tan^1_kQ)]^\perp$} \big\}.
$$

If the rank of ${\cal M}$ is maximal; i.e. ${\rm rank}\,{\cal M}=n$, 
there are no dynamical constraints and
the equations (\ref{aff-Leq}) give $n$ of the coefficients $(X_\alpha)^i$ as functions of the remaining ones. 
In fact, in this case we have
$$
[\fX^k(\Tan^1_kQ)]^\perp=\bigcap_{\alpha=1}^k\Ker\omega_\Lag^\alpha=\fX^V(\Tan^1_kQ) \ ;
$$
hence, for every $Z\in[\fX^k(\Tan^1_kQ)]^\perp$, 
the function $\inn(Z)\d\cE_\cL$ vanishes everywhere on $\Tan^1_kQ$
because $Z$ are $\pi$-vertical vector fields and 
$d\cE_\Lag$ is a $\pi$-horizontal form;
then there are no dynamical constraints and $P_1=\Tan^1_kQ$.

If ${\rm rank}\,{\cal M}<n$,
the elements of $[\fX^k(\Tan^1_kQ)]^\perp$ are 
$$
Z=Z^i\frac{\partial}{\partial\x^i}+
\sum_{\beta=1}^kZ^i_\beta\frac{\partial}{\partial\v^i_\beta}
\quad , \quad
\mbox{\rm with $\ds\left(\frac{\partial F^\alpha_i}{\partial\x^j}-\frac{\partial F^\alpha_j}{\partial\x^i}\right)Z^j=0$} \ ,
$$
and dynamical constraints may arise:
$$
\zeta^{(Z)}_1=\inn(Z)\d\cE_\Lag=\frac{\partial G}{\partial\x^i}Z^i \ .
$$

If ${\rm rank}\,{\cal M}=0$ ($M^\alpha_{ij}=0$), 
then $[\fX^k(\Tan^1_kQ)]^\perp=\fX(\Tan^1_kQ)$
and, hence, for each $\ds Z_i=\derpar{}{\x^i}$,
we get a dynamical constraint
\beq
\zeta^{(Z_i)}_1\equiv\zeta^i_1=\inn(Z_i)\d\cE_\Lag=\frac{\partial G}{\partial\x^i}\ .
\label{1stcon}
\eeq
Observe that these constraints are obtained also directly from the equation \eqref{aff-Leq}.

\noindent\textbf{First-generation non-dynamical constraints:}
Following the algorithm, for any Lagrangian $k$-vector field $\bX$ and any $\bY\in\fX^k(\Tan^1_kQ)$ such that $\bX+\bY$ is a \textsc{sopde}, we have
$$
S_1=\big\{x\in\Tan^1_kQ\ \big|\ \big(\inn(Z)\bOm_\Lag(\bY)\big)(x)=0,\ \mbox{\rm for every $Z\in\fM$} \big\} \ .
$$
For affine Lagrangians, as $\fX^{kV}(T^1_kQ)\subseteq\Ker\bOm_\Lag$, 
we have $\fM=\fX(\Tan^1_kQ)$. 
Since all vector fields are in $\fM$, the functions defining $S_1$ in $P_1$ 
are obtained simply imposing 
$\bOm_\Lag(\bY)=0$. For simplicity we can take
$$
Y_\alpha=\big[\v^i_\alpha-(X_\alpha)^i\big]\frac{\partial}{\partial\x^i} \ .
$$
Then the constraints read
$$
\bOm_\Lag(\bY)=\sum_{\alpha=1}^k\big[\v^i_\alpha-(X_\alpha)^i\big]M^\alpha_{ij}\d\x^j=0
\quad \Longleftrightarrow \quad
\sum_{\alpha=1}^k\big[\v^i_\alpha-(X_\alpha)^i\big]M^\alpha_{ij}=0\ .
$$
As $\bX$ satisfies (\ref{aff-Leq}), we have the following set of equations
\beq{}\label{eq:ndGC}
\eta^j_1\equiv\sum_{\alpha=1}^k\v^i_\alpha M^\alpha_{ij}+\frac{\partial G}{\partial\x^j}=0\quad , \quad \mbox{\rm (on $P_1$)} \ ;
\eeq{}
and, depending on the rank of ${\cal M}$, these equations give new constraints or not;
in particular:
\bit
\item
For $0<{\rm rank}\,{\cal M}\leq n$, these equations are not satisfied in the whole $P_1$ 
and they give \sopde first generation constraints that define the submanifold 
$S_1\hookrightarrow P_1$.
The number of independent constraints is fixed by
${\rm rank}\,{\cal M}$.
\item
For $M^\alpha_{ij}=0$ (${\rm rank}\,{\cal M}=0$),
bearing in mind \eqref{1stcon},
the equalities hold everywhere in $P_1$; 
so no \sopde constraints appear in this case and $S_1=P_1$.
\eit

Notice that \eqref{eq:ndGC} are just the field equations \eqref{aff-Leq}
for \sopde $k$-vector fields which,
when they are integrable (i.e., holonomic),
these equations are the Euler-Lagrange field equations.
This means that, for affine Lagrangians,
the field equations are recovered as constraints of the theory. This is in accordance with a similiar result in \cite{GR-2016}

\noindent\textbf{Tangency conditions and further generation of constraints:}
At this point, we have \sopde $k$-vector fields $\bG=(\Gamma_\alpha)\in\vf^k(\Tan^1_kQ)$ 
which are Euler--Lagrange $k$-vector field on $S_1$. 
Now, we impose tangency and thus,
for \ 
$\ds\Gamma_\alpha=\v^i_\alpha\frac{\partial}{\partial\x^i}+\sum_{\beta=1}^k(\Gamma_\alpha)^i_\beta\frac{\partial}{\partial\v^i_\beta}$:
\bit
\item
The tangency conditions on the first-generation dynamical constraints $\zeta^{(Z_k)}_1$ are
\beq
\zeta_2^{(\alpha;k)}\equiv\Gamma_\alpha(\zeta^{(Z_k)}_1)=0 
\quad , \quad \mbox{\rm (on $S_1$)} \ ;
\label{tancondyn}
\eeq
and, if these conditions do not hold on $S_1$, they give new constraints
that define the submanifold $S_2\hookrightarrow~S_1$.
\item
When $0<{\rm rank}\,M=r\leq n$
there are also first-generation non-dynamical constraints $\eta^k_1$ and
the tangency conditions for them  lead to the equations
$$
\Gamma_\alpha(\eta^k_1)=\v^i_\alpha\v^j_\gamma\frac{\partial M^\gamma_{jk}}{\partial\x^i}+\sum_{\beta=1}^k(\Gamma_\alpha)^i_\beta M^\beta_{ik}+\v^i_\alpha\frac{\partial^2G}{\partial\x^i\partial\x^k}=0
\quad , \quad \mbox{\rm (on $S_1$)} \ .
$$
As the number of independent first-generation constraints is equal to $r={\rm rank}\,M$,
this system is compatible and thus
the tangency conditions determine $r$ of the coefficients $(\Gamma_\alpha)^i_\beta$
as functions of the remaining ones, 
and no new constraints appear.
\eit

In the case that new constraints have arisen,
we continue demanding the tangency of the solutions until the algorithm stabilizes.

\subsubsection{An academical example}

Consider a manifold $Q$ with local coordinates $(\x^1,\x^2)$ and its 2-tangent bundle $\Tan_2^1Q$ with local coordinates $(\x^1,\x^2,\v^1_1,\v^2_1,\v^1_2,\v^2_2)$. We consider the Lagrangian
$$
\cL=\x^2\v_1^1-\x^1\v^2_2+\x^1\x^2.
$$
For this Lagrangian $F^2_1=F^1_2\equiv 0$, $F^1_1=\x^2$, $F^2_2=-\x^1$, and $G=\x^1\x^2$. We get
$$
\d\cE_\Lag=-\x^2d\x^1-\x^1d\x^2\quad , \quad \omega_\Lag^1=\omega_\Lag^2=d\x^1\wedge d\x^2.
$$
We also have
$$
{\cal M}\equiv\begin{pmatrix}
M^1_{11} & M^1_{21} & M^2_{11} & M^2_{21}\\
M^1_{12} & M^1_{22} & M^2_{12} & M^2_{22}\end{pmatrix}=\begin{pmatrix}
0 & -1 & 0 & -1\\
1 & 0 & 1 & 0
\end{pmatrix},
$$
which indeed has maximal rank, and $\ds\bigcap_{\alpha=1}^2\Ker\omega_\Lag^\alpha=\fX^{kV}(\Tan_2^1Q)$.
From the above general discussion for affine Lagrangians 
we know that no dynamical constraints appear. 
Indeed, the Lagrangian equation has solutions
$\bX=(X_1,X_2)\in\vf^2(\Tan^1_2Q)$ everywhere in $\Tan^1_2Q$
given by
$$
X_\alpha=(X_\alpha)^i\derpar{}{q^i}+
\sum_{\beta=1}^k(X_\alpha)^i_\beta\frac{\partial}{\partial\v^i_\beta}
\quad , \quad \mbox{\rm where \quad $(X_1)^2+(X_2)^2=\x^2,\ (X_1)^1+(X_2)^1=-\x^1$} \ .
$$
Equivalently, the algorithm tells us that the first generation dynamical constraints come from $\inn(Z)\d\cE_\Lag=0$,
for $\ds Z\in\bigcap_{\alpha=1}^2\Ker\omega_\Lag^\alpha$; 
but these equalities hold trivially because $Z$ are 
$\pi$-vertical vector fields and 
$d\cE_\Lag$ is a $\pi$-horizontal form.

However, there are \sopde constraints.
In fact, from \eqref{eq:ndGC},
$$
\eta^j_1\coloneqq\sum_{\alpha=1}^k\v^i_\alpha M^\alpha_{ij}+\frac{\partial G}{\partial\x^j}=\begin{pmatrix}
0 & -1 & 0 & -1\\
1 & 0 & 1 & 0
\end{pmatrix}\begin{pmatrix}
\v_1^1\\
\v_1^2\\
\v_2^1\\
\v_2^2
\end{pmatrix}+\begin{pmatrix}
\x^2 \\ \x^1
\end{pmatrix}=0 
$$
and we get the constraints
$$
\eta^1_1=-\v^2_1-\v^2_2+\x^2=0 \quad ,\quad \eta_1^2=\v_1^1+\v_2^1+\x^1=0 \ ,
$$
which define the submanifold $S_1$. 
We get no conditions on the coefficients of the solutions, and hence
we have that any \sopde is an Euler--Lagrange $2$-vector field on $S_1$.

Let us now impose tangency. According to the general discussion, 
it only gives conditions on the coefficients of the general solution 
$\dst\Gamma_\alpha=\v^\mu_\alpha\frac{\partial}{\partial\x^\mu}+
\sum_{\beta=1}^k(X_\alpha)^\mu_\beta\frac{\partial}{\partial\v^\mu_\beta}$.
Indeed,
\beann
\Gamma_1(\eta_1^1)=\big(\v_1^2-(X_1)^2_1-(X_1)^2_2\big)\vert_{S_1}=0 
& \Longleftrightarrow & (X_1)^2_2\vert_{S_1}=\v_1^2-(X_1)^2_1,
\\
\Gamma_1(\eta_2^1)=\big(\v_1^1+(\Gamma_1)^1_1+(\Gamma_1)^1_2\big)\vert_{S_1}=0
& \Longleftrightarrow & (X_1)^1_2\vert_{S_1}=-\v_1^1-(X_1)^1_1,
\\
\Gamma_2(\eta_1^1)=\big(\v_2^2-(X_2)^2_1-(X_2)^2_2\big)\vert_{S_1}=0
& \Longleftrightarrow & (X_2)^2_2\vert_{S_1}=\v_2^2-(X_2)^2_1,
\\
\Gamma_2(\eta_2^1)=\big(\v_2^1+(X_2)^1_1+(X_2)^1_2\big)\vert_{S_1}=0
& \Longleftrightarrow &
(X_2)^1_2\vert_{S_1}=-\v_2^1-(X_2)^1_1 \ .
\eeann
Summing up, we can find Euler--Lagrange $2$-vector fields only on
$$
S_1=\big\{x\in\Tan^1_2Q\ \big|\ \x^1(x)=-\v_1^1(x)-\v_2^1(x),\ \x^2(x)=+\v^2_1(x)+\v^2_2(x)\big\}
$$
which are $\bG=(\Gamma_1,\Gamma_2)$ with
$$
\Gamma_1\big|_{S_1}=\v^1_1\frac{\partial}{\partial\x^1}+\v^2_1\frac{\partial}{\partial\x^2}+A\frac{\partial}{\partial\v^1_1}+B\frac{\partial}{\partial\v^2_1}-(\v_1^1+A)\frac{\partial}{\partial\v^1_2}+(\v^2_1-B)\frac{\partial}{\partial\v^2_2},
$$
$$
\Gamma_2\big|_{S_1}=\v^1_2\frac{\partial}{\partial\x^1}+\v^2_2\frac{\partial}{\partial\x^2}+C\frac{\partial}{\partial\v^1_1}+D\frac{\partial}{\partial\v^2_1}-(\v_2^1+C)\frac{\partial}{\partial\v^1_2}+(\v^2_2-D)\frac{\partial}{\partial\v^2_2},
$$
where $A,B,C,D$ are free functions. 
If we want this family of solutions to be integrable we must impose
$$
[\Gamma_1,\Gamma_2]\vert_{S_1}=0 \ ,
$$
and this equation leads to a system of partial differential equations on the functions $A,B,C,D$.


\subsection{The Einstein-Palatini model}

The Einstein-Palatini model is a first order singular field theory. 
A multisymplectic formulation of the model has been developed
in several works (see, for instance \cite{Capriotti, Capriotti2,MRosado-2013,GR-2018}).
In particular, in \cite{GR-2018} the constraints arising from the application
of the constraint algorithm in the (pre)multisymplectic framework 
have been computed explicitly in coordinates. 
They have a diverse origin and characteristics, which makes this system an interesting test for the theory developed in this article. Moreover, we provide an intrinsic characterization of the constraints.

The configuration bundle for this system is the bundle $\pi\colon{\rm E}\rightarrow M$, 
where $M$ is a connected orientable 4-dimensional manifold representing space-time, 
whose volume form is denoted $\eta\in\df^4(M)$, and
${\rm E}=\Sigma\times_MC(LM)$, where $\Sigma$ is the manifold of Lorentzian metrics on $M$,with signature $(-+++)$, 
and $C(LM)$ is the bundle of connections on $M$;
that is, linear connections in $\Tan M$.

Consider natural coordinates $(x^\mu)$ in $M$ such that 
$\eta=\d x^0\wedge\ldots\wedge\d x^3\equiv\d^4x$. 
We use adapted fiber coordinates in ${\rm E}$,
denoted $(x^\mu,g_{\alpha\beta},\Gamma^\nu_{\lambda\gamma})$,
and then the coordinates in $\Tan^1_kQ$ are
$(x^\mu,g_{\alpha\beta},\Gamma^\nu_{\lambda\gamma},g_{\alpha\beta,\mu},\Gamma^\nu_{\lambda\gamma,\mu})$
(with $\mu,\nu,\gamma,\lambda=0,1,2,3$
and $0\leq\alpha\leq\beta\leq 3$, 
where this ordering rules for sum over symmetric indices and not over all the components). 
The functions $g_{\alpha\beta}$ are the components of the metric associated to the charts in the base $(x^\mu)$,
and $\Gamma^\nu_{\lambda\gamma}$ are the Christoffel symbols of the connection. 
We do not assume torsionless connections and hence
$\Gamma^\nu_{\lambda\gamma}\neq\Gamma^\nu_{\gamma\lambda}$,
in general.

The Einstein-Palatini Lagrangian is 
$$
L_{\rm EP}=\sqrt{|{\rm det}(g)|}\,g^{\alpha\beta}R_{\alpha\beta}\equiv
\varrho g^{\alpha\beta}R_{\alpha\beta}=\varrho\,R \ ,
$$
where $\varrho=\sqrt{|det(g_{\alpha\beta})|}$, 
$R=g^{\alpha\beta}R_{\alpha\beta}$ is the {\sl scalar curvature},
$R_{\alpha\beta}=
\Gamma^{\gamma}_{\beta\alpha,\gamma}-\Gamma^{\gamma}_{\gamma\alpha,\beta}+
\Gamma^{\gamma}_{\beta\alpha}\Gamma^{\sigma}_{\sigma\gamma}-
\Gamma^{\gamma}_{\beta\sigma}\Gamma^{\sigma}_{\gamma\alpha}$
are the components of the {\sl Ricci tensor}, which depend only on the connection, and
$g^{\alpha\beta}$ denotes the inverse matrix of $g$, 
namely: $g^{\alpha\beta}g_{\beta\gamma}=\delta^\alpha_\gamma$.

The Einstein-Palatini Lagrangian is affine, thus
we will follow the guidelines developed in Section \ref{subsubsect:GC} 
for these kinds of Lagrangians to streamline the process. 
In this case, the functions  $F^\alpha_i$ and $G$ appearing in \eqref{aflag} are
$$
F^{\beta\gamma,\mu}_\alpha=
\varrho(\delta_\alpha^\mu g^{\beta\gamma}-\delta_\alpha^\beta g^{\mu\gamma})
\quad , \quad 
G=\varrho g^{\alpha\beta}(\Gamma^{\gamma}_{\beta\alpha}\Gamma^{\sigma}_{\sigma\gamma}-
\Gamma^{\gamma}_{\beta\sigma}\Gamma^{\sigma}_{\gamma\alpha})\ .
$$
Notice that we only have components on the fiber coordinates 
$\Gamma^\nu_{\lambda\gamma,\mu}$, but not on $g_{\alpha\beta,\mu}$.
Therefore, the matrix $\ds{\cal M}=\left(\frac{\partial F^\alpha_i}{\partial\x^j}-\frac{\partial F^\alpha_j}{\partial\x^i} \right)$
introduced in Section  \ref{subsubsect:GC}
has constant rank, $0<{\rm rank}\,{\cal M}<n$.
Then, we have that
\beann
\cE_\cL=-G
\quad &\Longrightarrow&\quad
\d\cE_\cL=-\d G ,
\\ 
\theta_{L_{\rm EP}}^\mu=F^{\beta\gamma,\mu}_\alpha\,\d\ \Gamma^{\alpha}_{\beta\gamma}
\quad &\Longrightarrow&\quad
\omega^\mu_{L_{\rm EP}}=
-\d F^{\beta\gamma,\mu}_{\alpha}\wedge\d \Gamma^{\alpha}_{\beta\gamma}\ ;
\eeann
and the solutions to the Einstein equations in the Einstein-Palatini approach
are the integral sections of holonomic $k$-vector fields
$\bX=(X_\alpha)\in\vf^k(\Tan^1_kQ)$  such that
\beq
\sum_{\alpha=1}^k\inn(X_\alpha)\omega^\alpha_{L_{\rm EP}}=\d\cE_{L_{\rm EP}}\ .
\label{EPg}
\eeq
Writing
$$
X_\alpha=
\sum_{\sigma\leq\rho}\left( f_{\sigma\rho,\alpha}
\frac{\partial}{\partial g_{\sigma\rho}}+
f_{\sigma\rho\mu,\alpha}\frac{\partial}{\partial g_{\sigma\rho,\mu}}\right)
+f_{\beta\gamma,\alpha}^\nu
\frac{\partial}{\partial \Gamma^\nu_{\beta\gamma}}+
f^\nu_{\beta\gamma\mu,\alpha}\frac{\partial}{\partial \Gamma^\nu_{\beta\gamma,\mu}}\ ,
$$
the above equations are
\beq
\label{EP-EL}
-\Gamma^\alpha_{\beta\gamma,\mu} \frac{\partial{F^{\beta\gamma,\mu}_\alpha}}{\partial g_{\rho\sigma}}+\frac{\partial G}{\partial g_{\rho\sigma}}=0
\quad ; \quad
\sum_{\rho\leq\sigma}g_{\rho\sigma,\mu} \frac{\partial{F^{\beta\gamma,\mu}_\alpha}}{\partial g_{\rho\sigma}}+\frac{\partial G}{\partial \Gamma^\alpha_{\beta\gamma}}=0
\ .
\eeq

\noindent\textbf{First-generation dynamical constraints:}
Following the algorithm, generic $k$-vector fields solution to the equations \eqref{EPg} exist only on the points of the submanifold
$$
P_1=\big\{x\in\Tan^1_kQ\ \big|\ (\inn(Z)\d\cE_\cL)(x)=0,\ \mbox{\rm for every $Z\in[\fX^k(\Tan^1_kQ)]^\perp$} \big\} \ .
$$
Then, for a generic vector field 
$$
Z=
\sum_{\rho\leq\sigma}\left(f_{\rho\sigma}\frac{\partial}{\partial g_{\rho\sigma}}+
f_{\rho\sigma\mu}\frac{\partial}{\partial g_{\rho\sigma,\mu}}\right)+f^\alpha_{\beta\gamma}
\frac{\partial}{\partial \Gamma^\alpha_{\beta\gamma}}+
f^\alpha_{\beta\gamma\mu}\frac{\partial}{\partial \Gamma^\alpha_{\beta\gamma,\mu}}\ .
$$
the condition that
$Z\in[\fX^k(\Tan^1_kQ)]^\perp$ holds if, and only if,
$$
f_{\rho\sigma}\frac{\partial{F^{\beta\gamma,\mu}_\alpha}}{\partial g_{\rho\sigma}}=0\quad ,  \quad
f_{\beta\gamma}^\alpha\frac{\partial{F^{\beta\gamma,\mu}_\alpha}}{\partial g_{\rho\sigma}}=0\quad ;
\quad \mbox{\rm 
(for \ $0\leq\rho<\sigma\leq3$, $0\leq\alpha,\beta,\gamma\leq3$)} \ .
$$
These equations appear in \cite{GR-2018} (Proposition 3.7),
and their solutions are obtained giving
$f_{\rho\sigma}=0$ and $f^\alpha_{\beta\gamma}=C_\beta\delta^\alpha_\gamma+K^\alpha_{\beta\gamma}$, for functions $C_\beta$ and $K^\alpha_{\beta\gamma}$ such that $K^\alpha_{\alpha\gamma}=0$ and $K^\alpha_{\beta\gamma}+K^\alpha_{\gamma\beta}=0$. 
These last conditions can be rewritten in a more suitable way as follows: 
these functions $K^\alpha_{\beta\gamma}$ are a linear combination of the new functions
$$
S^\alpha_{\beta\gamma,\lambda\rho\nu}=\frac13g_{\lambda\nu}g_{\rho\beta}\delta^\alpha_\gamma-\frac13g_{\rho\nu}g_{\lambda\beta}\delta^\alpha_\gamma+\frac13g_{\rho\nu}g_{\lambda\gamma}\delta^\alpha_\beta-\frac13g_{\lambda\nu}g_{\rho\gamma}\delta^\alpha_\beta+g_{\lambda\beta}g_{\rho\gamma}\delta^\alpha_\nu-g_{\rho\beta}g_{\lambda\gamma}\delta^\alpha_\nu \ .
$$
Indeed, $S^\alpha_{\beta\gamma,\lambda\rho\nu}$ satisfies the conditions for any $\lambda,\rho$ and $\nu$, and $K^\alpha_{\beta\gamma}=\frac12K^\nu_{\sigma\tau}g^{\lambda\sigma}g^{\rho\tau} S^\alpha_{\beta\gamma,\lambda\rho\nu}$.

The coefficients $f_{\rho\sigma\mu}$ and $f^\alpha_{\beta\gamma\mu}$ 
do not contribute to originate constraints because the Lagrangian is affine. Moreover, $\inn\Big(C_\beta\delta^\alpha_\gamma\derpar{}{\Gamma^\alpha_{\beta\gamma}}\Big)\d\cE_\cL=0$, 
and nor do they produce constraints. Finally,
$$
\inn\Big(S^\alpha_{\beta\gamma,\lambda\rho\nu}\frac{\partial }{\partial \Gamma^\alpha_{\beta\gamma}}\Big)\d\cE_\cL=
g_{\lambda\mu}T^\mu_{\rho\nu}-g_{\rho\mu}T^\mu_{\lambda\nu}+\tfrac{1}{3}g_{\lambda\nu}T^\mu_{\mu\rho}-
\tfrac{1}{3}g_{\rho\nu}T^\mu_{\mu \lambda}\ ,
$$
where $T^\alpha_{\beta\gamma}$ are the components of the torsion tensor
which are defined as usual,
$T^\alpha_{\beta\gamma}=\Gamma^\alpha_{\beta\gamma}-\Gamma^\alpha_{\gamma\beta}$.
These are the first-generation dynamical constraints defining $P_1$,
which can be written in an equivalent way as
$$
(\zeta_1)_{\beta\gamma}^\alpha\equiv
T^\alpha_{\beta\gamma}-
\frac13\delta^\alpha_\beta T^\nu_{\nu\gamma}+\frac13\delta^\alpha_\gamma T^\nu_{\nu\beta}=0\ 
$$
and they are called  {\sl torsion constraints} in \cite{GR-2018}.

The general solution to the equation \eqref{EPg}
(before demanding the \sopde condition) are $k$-vector fields
$\mathbf{X}=(X_\nu)\in\vf^k(\Tan^1_kQ)$ with
\bea
\nonumber
X_\alpha&=&
\sum_{\sigma\leq\rho}\left(\Big(
g_{\sigma\lambda}\Gamma^\lambda_{\alpha\rho}+g_{\rho\lambda}\Gamma^\lambda_{\alpha\sigma}+
\frac{2}{3}g_{\sigma\rho}T^\lambda_{\lambda\nu}\Big)
\frac{\partial}{\partial g_{\sigma\rho}}+
f_{\sigma\rho\mu,\alpha}\frac{\partial}{\partial g_{\sigma\rho,\mu}}\right)
 \\ & &
+\Big(\Gamma^\lambda_{\alpha\gamma}\Gamma^\nu_{\beta\lambda}+
C^\nu_{\beta\gamma,\alpha}+K^\nu_{\beta\gamma,\alpha}\Big)
\frac{\partial}{\partial \Gamma^\nu_{\beta\gamma}}+
f^\nu_{\beta\gamma\mu,\alpha}\frac{\partial}{\partial \Gamma^\nu_{\beta\gamma,\mu}}
\quad ; \quad \mbox{\rm (on $P_1$)} \ ,
\label{mvfsolgen}
\eea
for some functions 
$C^\nu_{\beta\gamma,\mu}, K^\nu_{\beta\gamma,\mu}\in C^\infty(\Tan^1_kQ)$ satisfying that
$$
C^\nu_{\beta\gamma,\mu}=C_{\beta\mu}\delta^\nu_\gamma \quad,  \quad
K^\nu_{\nu\gamma\mu}=0 \quad , \quad
K^{\nu}_{\beta\gamma \nu}+K^\nu_{\gamma\beta \nu}=0 
\quad ; \quad \mbox{\rm (on $P_1$)} \  .
$$

\noindent\textbf{First-generation non-dynamical constraints:}
Following the algorithm, for every Lagrangian $k$-vector field $\bX$
given by \eqref{mvfsolgen},
we can take any $\bY=(Y_\nu)\in\fX^k(\Tan^1_kQ)$ 
such that $\bG=\bX+\bY$ is a \textsc{sopde}.
The simplest choice is
$$
Y_\alpha=\sum_{\sigma\leq\rho}
\Big(g_{\sigma\rho,\alpha}-g_{\sigma\lambda}\Gamma^\lambda_{\alpha\rho}-g_{\rho\lambda}\Gamma^\lambda_{\alpha\sigma}-
\frac{2}{3}g_{\sigma\rho}T^\lambda_{\lambda\alpha}\Big)
\frac{\partial}{\partial g_{\sigma\rho}}+
\Big(\Gamma^\nu_{\beta\gamma,\alpha}-\Gamma^\lambda_{\alpha\gamma}\Gamma^\nu_{\beta\lambda}-
C^\nu_{\beta\gamma,\alpha}-K^\nu_{\beta\gamma,\alpha}\Big)
\frac{\partial}{\partial \Gamma^\nu_{\beta\gamma}}\,.
$$
Then we have
$$
S_1=\big\{x\in\Tan^1_kQ\ \big|\ \big(\inn(Z)\bOm_\Lag(\bY)\big)(x)=0,\ \mbox{\rm for every $Z\in\fM$} \big\} \ ;
$$
where $\fM=\fX(\Tan^1_kQ)$.
For an affine Lagrangian this conditions is realized by \eqref{eq:ndGC} which, in this case, takes the form
\bea
\label{1ndc}
(\eta_1)^{\rho\sigma}&\equiv&-\Gamma^\alpha_{\beta\gamma,\mu} \frac{\partial{F^{\beta\gamma,\mu}_\alpha}}{\partial g_{\rho\sigma}}+\frac{\partial G}{\partial g_{\rho\sigma}}=0\ ;
\\
\label{2ndc}
(\eta_1)^{\beta\gamma}_{\alpha}&\equiv&\sum_{\rho\leq\sigma}g_{\rho\sigma,\mu} \frac{\partial{F^{\beta\gamma,\mu}_\alpha}}{\partial g_{\rho\sigma}}+\frac{\partial G}{\partial \Gamma^\alpha_{\beta\gamma}}=0\ .
\eea
As it was pointed out in the general case of affine Lagrangians,
these equations are just the field equations \eqref{EP-EL}
for \sopde $k$-vector fields and
the field equations are recovered as constraints of the theory.
For holonomic $k$-vector fields these equations
lead to the {\sl Einstein equations}.

The equations \eqref{1ndc} are the so-called {\sl connection constraints} in \cite{GR-2018}. 
The equations \eqref{2ndc},
called {\sl metric constrains} in \cite{GR-2018},
after taking into account the first-generation dynamical constraints can
be rewritten as
$$
(\eta_1)^{\beta\gamma}_{\alpha}\equiv - g_{\rho\sigma,\mu}+
g_{\sigma\lambda}\Gamma^\lambda_{\mu\rho}+
g_{\rho\lambda}\Gamma^\lambda_{\mu\sigma}+
\frac{2}{3}g_{\rho\sigma}T^\lambda_{\lambda\mu} \ ,
$$
and are then called {\sl pre-metricity constraints} in \cite{GR-2018}.
These constraints define
the submanifold $S_1\hookrightarrow~P_1$.

\noindent\textbf{Tangency conditions and further generation of constraints:}
The second generation constraints arise from the tangency condition of $\bG$ on 
the constraints $(\zeta_1)_{\beta\gamma}^\alpha$, and they are
$$
(\eta_2)^\alpha_{\beta\gamma,\nu}\equiv T^\alpha_{\beta\gamma,\nu}-
\frac13\delta^\alpha_\beta T^\mu_{\mu\gamma,\nu}+
\frac13\delta^\alpha_\gamma T^\mu_{\mu\beta,\nu}=0\ ;
$$
which are second generation \sopde constraints,
since $Y_\alpha((\zeta_1)_{\lambda\rho\nu})\not=0$.
The tangency condition on the other constraints does not lead to new constraints
(as they are non-dynamical constraints).
Thus we have obtained a new submanifold $S_2\hookrightarrow S_1$.

Finally, the tangency condition on the new constraints
$(\eta_2)^\alpha_{\beta\gamma,\nu}$ 
does not lead to new constraints because
they are also non-dynamical constraints.
In fact, this condition reads
$$
X_\lambda\Big(T^\alpha_{\beta\gamma,\nu}-
\frac13\delta^\alpha_\beta T^\mu_{\mu\gamma,\nu}+
\frac13\delta^\alpha_\gamma T^\mu_{\mu\beta,\nu}\Big)=
f^\alpha_{\beta\gamma\nu,\lambda}-\frac13\delta^\alpha_\beta f^\mu_{\mu\gamma\nu,\lambda}+
\frac13\delta^\alpha_\gamma f^\mu_{\mu\beta\nu,\lambda}=0
\ ; \ \mbox{\rm (on $S_2$)}  \ ,
$$
which are equations for the functions $f^\alpha_{\beta\gamma\mu,\lambda}$.

Hence $S_2$ is the final constraint submanifold 
$S_f\hookrightarrow\Tan^1_kQ$,
which is defined in $\Tan^1_kQ$ by the set of constraints
$$
(\zeta_1)_{\lambda\zeta\nu}=0 \quad , \quad
(\eta_1)^{\rho\sigma}=0 \quad , \quad
(\eta_1)^{\beta\gamma}_{\alpha}=0 \quad , \quad
(\eta_2)^\alpha_{\beta\gamma,\nu}=0 \ .
$$

\noindent\textbf{Integrability conditions:}
By completeness, the integrability conditions for the Einstein-Palatini model are
obtained by imposing that $[X_\alpha,X_\beta]=0$ (on $S_f$),
and they are
$$g_{\rho\gamma}\Gamma^\gamma_{[\nu\lambda}\Gamma^\lambda_{\mu]\sigma}+g_{\sigma\gamma}\Gamma^\gamma_{[\nu\lambda}\Gamma^\lambda_{\mu]\rho}+g_{\rho\lambda}\Gamma^\lambda_{[\mu\sigma,\nu]}+g_{\sigma\lambda}\Gamma^\lambda_{[\mu\rho,\nu]}+\frac23g_{\rho\sigma}T^\lambda_{\lambda[\mu,\nu]}=0\ .
$$
These are new {\sl integrability constraints} that define a new submanifold 
${\cal S}_f\hookrightarrow S_f$ where there are integral sections which are solution to the Einstein equations.
It can be checked that the tangency condition of $\bG$ holds on ${\cal S}_f$.
(See  \cite{GR-2018} for the details).


\section{Conclusions and outlook}

The first goal of this work has been to solve the second-order problem for singular field theories, 
generalizing the constraint algorithm of  \cite{MR} (for singular Lagrangian mechanical systems) 
to $k$-presymplectic Lagrangian systems and completing the results of \cite{GMR}. 
In particular, given a  $k$-symplectic Lagrangian system $(\Tan^1_kQ;\cL)$, 
we have developed an algorithm that produces the maximal submanifold of $\Tan^1_kQ$ 
on which the Lagrangian equation has solution and the solution can be chosen to be a \textsc{sopde}. 

The algorithm works as follows: 
First, we characterize the submanifold $P_1$
on which the Lagrangian equation has a solution. 
The constraints defining $P_1$ are called 
{\sl $k$-presymplectic} or {\sl dynamical constraints},
they arise from demanding the compatibility condition,
and can be chosen to be $\Leg$-projectable functions.
Second, the algorithm gives the submanifold $S_1$ of $P_1$ on which the Lagrangian $k$-vector fields can be chosen to be \textsc{sopde}s. 
This produces new constraints that are called {\sl \sopde} or {\sl non-dynamical constraints}. 
They arise from demanding the \sopde condition and are not $\Leg$-projectable. 

Next, the stability or tangency condition is imposed,
looking for the submanifold $S_2$ of $S_1$ 
where \sopde solutions can be chosen to be tangent to $S_1$.
Then, the tangency condition on the non-dynamical constraints
give no new constraints, but on the dynamical constraints
it can produce new constraints that can be classified as dynamical or non-dynamical,
depending on their nature: if they are related to demanding solutions to be \sopde $k$-vector fields
then they are non-dynamical and they are dynamical otherwise.
This last step is repeated until the algorithm stabilizes; that is, 
in the most favourable cases, until we have a submanifold $S_f$ 
on which the \sopde $k$-vector fields 
solutions to the $k$-presymplectic Lagrangian equations are tangent to it.
In each step, results similar to the previous ones are repeated.
In this way, the behaviour of the constraint algorithm 
for $k$-presymplectic Lagrangian systems in field theory
is exactly the same as for presymplectic Lagrangian systems in mechanics
\cite{BGPR,BGPR2,CLR,CLR2,MR}.

Finally, the $\Leg$-projectability and the integrability 
of these \sopde $k$-vector fields are additional conditions to be demanded
that can produce new constraints.

As a very interesting case, we have applied the algorithm 
to analyze the {\sl Einstein-Palatini} or {\sl metric-affine model} 
of General Relativity, which is very suitable to be studied 
using the $k$-symplectic formulation;
and we have compared the results achieved with those obtained
in the multisymplectic analysis of this model.
As it is an affine Lagrangian, we have analyzed previously
the general case of affine Lagrangians in field theory and, as a particular case, 
we have described an academical example.

As further research,
this constraint analysis can be implemented to analyze other 
extended models of General Relativity,
such as Lovelock, $f(R)$ or $f(T)$ theories.
It should be also interesting to do a similar study
of the second-order problem for the multisymplectic formulation
of classical field theories in general.


\subsection*{Acknowledgments}

We acknowledge the financial support from the Spanish
Ministerio de Ciencia, Innovaci\'on y Universidades project
PGC2018-098265-B-C33
and the Secretary of University and Research of the Ministry of Business and Knowledge of
the Catalan Government project 2017--SGR--932.


\addcontentsline{toc}{section}{\bf References}
\itemsep 0pt plus 1pt

\end{document}